\documentclass[11pt]{article}
\usepackage[margin=1in]{geometry}
\usepackage{amsmath,amsthm,amssymb}
\usepackage{graphicx}
\usepackage{tikz}
\usepackage{todonotes}
\usepackage[pagebackref=true]{hyperref}
\usepackage{hyperref}

\usepackage[T1]{fontenc}
\usepackage{lmodern}

\newtheorem{theorem}{Theorem}
\newtheorem*{theorem*}{Theorem}
\newtheorem{corollary}[theorem]{Corollary}

\newtheorem*{conjecture*}{Conjecture}
\newtheorem{lemma}[theorem]{Lemma}

\newtheorem*{claim*}{Claim}

\theoremstyle{definition}
\newtheorem{definition}[theorem]{Definition}
\newtheorem{assumption}[theorem]{Assumption}

\theoremstyle{remark}
\newtheorem*{remark}{Remark}

\newcommand{\mech}{\text{\sc Mech}}
\newcommand{\opt}{\text{\sc Opt}}

\newcommand{\R}{\mathbb{R}}

\newcommand{\A}{\Theta}
\newcommand{\vol}{\text{vol}}

\newcounter{note}[section]

\title{On the Nisan-Ronen conjecture}

\author{George Christodoulou\thanks{Department of Computer Science,
    University of Liverpool, UK.  Email:
    \texttt{gchristo@liverpool.ac.uk} } \and Elias Koutsoupias \thanks{Department of Computer Science, University
    of Oxford, UK. Email: \texttt{elias@cs.ox.ac.uk}} \and Annam{\'a}ria
  Kov{\'a}cs\thanks{Department of Informatics, Goethe University,
    Frankfurt M., Germany. Email: \texttt{panni@cs.uni-frankfurt.de}}} 

\date{}

\begin{document}

\maketitle

\begin{abstract}
  The Nisan-Ronen conjecture states that no truthful mechanism for makespan-minimization when allocating $m$ tasks to $n$ unrelated machines can have approximation ratio less than $n$. Over more than two decades since its formulation, little progress has been made in resolving it and the best known lower bound is still a small constant. This work makes progress towards validating the conjecture by showing a lower bound of $1+\sqrt{n-1}$.
\end{abstract}

\section{Introduction}

This work makes progress on one of the most important open problems of algorithmic mechanism design~\cite{NRTV07}, the Nisan-Ronen conjecture.  Mechanism design, a celebrated branch of Game Theory and Microeconomics, studies a special class of algorithms, called mechanisms, which are robust under selfish behavior and produce a social outcome with a certain guaranteed quality. Unlike traditional algorithms that get their input from a single user, mechanisms solicit the input from different participants (called agents, players, bidders), in the form of preferences over the possible outputs (outcomes). The challenge stems from the fact that the actual preferences of the participants are private information, unknown to the algorithm. The participants are assumed to be utility maximisers that will provide some input that suits their objective and may differ from their true preferences. A {\em truthful} mechanism provides incentives such that a truthful input is the best action for each participant.

The question is what kind of problems can be solved within this framework. In their seminal paper that launched the field of algorithmic mechanism design, Nisan and Ronen~\cite{NR01} proposed the scheduling problem on unrelated machines as a central problem to capture the algorithmic aspects of mechanism design. In the classical form of the problem, which has been extensively studied from the algorithmic perspective, there are $n$ machines that process a set of $m$ tasks; each machine $i$ takes time $t_{ij}$ to process task $j$. The objective of the algorithm is to allocate each task to a machine in order to minimize the makespan, i.e., the maximum completion time of all machines. In the mechanism design setting, each machine provides as input its processing times for each task. The selection by Nisan and Ronen of this version of the scheduling problem to study the limitations that truthfulness imposes on algorithm design was a masterstroke, because it turned out to be an extremely rich and challenging setting.

Nisan and Ronen applied the VCG mechanism~\cite{Vic61,Cla71,Gro73}, the most successful generic machinery in mechanism design, which truthfully implements the outcome that maximizes the social welfare. In the case of scheduling, the allocation of the VCG is the greedy allocation: each task is independently assigned to the machine with minimum processing time. This mechanism is truthful, but has poor approximation ratio, $n$. They boldly conjectured that this is the best guarantee that can be achieved by any deterministic (polynomial-time or not) truthful mechanism.

\begin{conjecture*}[Nisan-Ronen]
  There is no deterministic truthful mechanism with approximation ratio better than $n$ for the problem of scheduling $n$ unrelated machines.
\end{conjecture*}

The conjectured bound refers only to the limitation that truthfulness imposes and does not preclude any computational limitations. In that sense, this is a very strong information-theoretic bound which should hold for all deterministic mechanisms, regardless of their running time.

The Nisan-Ronen conjecture is perhaps the most famous problem and, arguably, one of the most important open ones in algorithmic mechanism design. Despite intensive efforts, very sparse progress has been made towards its resolution. All the known lower bounds are small constants; Nisan and Ronen originally showed that no truthful deterministic mechanism can achieve an approximation ratio better than $2$. This was improved to $2.41$~\cite{ChrKouVid09}, and later to $2.61$ \cite{KV07}, which was the current best for over a decade. Very recently the bound was improved to $2.755$ by Giannakopoulos, Hamerl and Po\c{c}as~\cite{giannakopoulos2020} and then to $3$ by Dobzinski and Shaulker~\cite{DS20}. Still the gap from the best known upper bound of $n$ remains huge.

In this work we take a solid step towards the resolution of the conjecture by providing a drastic improvement on the lower bound.

\begin{theorem*}
  There is no deterministic truthful mechanism with approximation ratio better than $1+\sqrt{n-1}$ for the problem of scheduling $n$ unrelated machines.
\end{theorem*}

The proof of the theorem is presented in Section~\ref{sec:lb}. At its core, it is based on the use of a known characterization of $2\times 2$ instances, with only two machines and 2 tasks~\cite{CKK20}. Roughly speaking, this characterization says essentially that the only available algorithms for $2\times 2$ instances are either task independent (algorithms that allocate each task independently of the values of the other tasks) or affine minimizers (algorithms that compare simple linear expressions of the input values to allocate the tasks). This is an extremely limited class of algorithms, a fact that provides some weak evidence that the Nisan-Ronen conjecture may be true. No such characterization is known for 3 or more machines, and a widely-held belief has been that attacking the Nisan-Ronen conjecture requires further progress in this direction of characterizing truthful mechanisms for $n$ machines. It is therefore quite surprising that one can obtain a lower bound of $1+\sqrt{n-1}$ based mainly on the characterization of $2\times 2$ instances.

Another cornerstone of the proof of the main theorem is the extensive use of weak monotonicity, the characterizing property of truthful allocation functions. It is well known~\cite{SY05,AK08} that a mechanism is truthful if its allocation function is monotone in the values of each machine. Weak monotonicity in one dimension (i.e., a single task) is the usual notion of monotonicity of the allocation function, and for two or more dimensions, it takes a particular very natural form. Thus one can restate the main theorem as ``no monotone algorithm, polynomial-time or not, has approximation ratio less than $1+\sqrt{n-1}$ for the problem of scheduling $n$ unrelated machines''. In contrast, the approximation ratio for the usual (non-monotone) class of algorithms is trivially $1$, for exponential-time algorithms, and 2 for polynomial-time ones~\cite{lenstra1990approximation}. To use weak monotonicity, we consider perturbations of an input that lie in a small open hyperrectangle. The precise definition of these perturbations needs to balance different requirements of the proof, but we believe that the general idea can be useful in other similar situations.

Previous attempts for resolving the Nisan-Ronen conjecture employed a characterization of $2\times 2$ instances and, of course, weak monotonicity. What distinguishes this work from these previous attempts, with the exception of~\cite{CKK20}, is the use of particular inputs with some very high values that force any reasonable algorithm to allocate each task to only two machines. These inputs avoid the problems that weak monotonicity presents for multiple players: with such inputs, when one machine does not take a task, we know precisely which machine takes it.

Some ideas of this work, such as the use of the $2\times 2$ characterization on particular inputs, first appeared in a recent publication~\cite{CKK20}. The main result in~\cite{CKK20} is a lower bound of $\sqrt{n-1}$ for all deterministic truthful mechanisms, when the cost of processing a subset of tasks is given by a submodular (or supermodular) set function, instead of an additive function of the standard scheduling setting.

The instances employed in~\cite{CKK20} for submodular costs use only $n-1$ \emph{pairs of tasks}. But here, where costs are additive, we use exponentially many tasks. We don't know whether this large number of tasks is needed in general, but it is essential in our proof. It definitely needs to be $\Omega(n^{3/2})$ for the type of instances that we use, since for fewer tasks we know truthful mechanisms with approximation ratio less than $\sqrt{n-1}$~\cite{CKK21}.

The real challenge, not present in~\cite{CKK20}, is that the characterization of $2\times 2$ truthful mechanisms for additive costs includes mechanisms which are difficult to handle (for example, relaxed task independent mechanisms and relaxed affine minimizers). This obstacle adds many complications and requires a different proof structure.

The major open problem left open is to settle the Nisan-Ronen
conjecture. The techniques of this work may be helpful in this
direction. The case of randomized or fractional mechanisms appears to
be more challenging; the best known lower bound of the approximation
ratio is 2~\cite{MualemS18,CKK10}, embarrassingly lower than the best
known upper bound $(n+1)/2$~\cite{CKK10}. The bottleneck of applying
the techniques of the current work to these variants appears to be the
lack of a good characterization of $2\times 2$ mechanisms.  Finally,
although the result of this work indicates that mechanisms constitute
a limited subclass of allocation algorithms, a more direct
demonstration would be to find other useful properties of mechanisms,
or even obtain a more global characterization of mechanisms for the domain
of scheduling and its generalizations.

\subsection{Further related work}
The lack of progress in the original unrelated machine problem led to the study of
variants and special cases, for which significant results have been obtained. Ashlagi
et al.~\cite{ADL09}, resolved a restricted version of the Nisan-Ronen conjecture, for
the special but natural class of {\em anonymous} mechanisms. Lavi and
Swamy~\cite{LaviS09} studied a restricted input domain which however retains the
multi-dimensional flavour of the setting. They considered inputs with only two
possible values ``low'' and ``high'', that are publicly known to the designer of the
algorithm. For this case they showed an elegant deterministic mechanism with an
approximation factor of 2. They also showed that even for this setting achieving the
optimal makespan is not possible under truthfulness, and provided a lower bound of
$11/10$. Yu~\cite{Yu09} extended the results for a range of values, and Auletta et
al.~\cite{Auletta0P15} studied multi-dimensional domains where the private
information of the machines is a single bit.

Randomized mechanisms have also been studied and have slightly
improved guarantees. There are two notions of truthfulness for
randomized mechanisms; a mechanism is {\em universally truthful} if it
is defined as a probability distribution over deterministic truthful
mechanisms, and it is {\em truthful-in-expectation}, if in expectation
no player can benefit by lying. In \cite{NR01}, a universally truthful
mechanism was proposed for the case of two machines, and was later
extended to the case of $n$ machines by Mu'alem and
Schapira~\cite{MualemS18} with an approximation guarantee of $0.875n$,
which was later improved to $0.837n$ by \cite{LuYu08}. Lu and
Yu~\cite{LuY08a} showed a truthful-in-expectation mechanism with an
approximation guarantee of $(n+5)/2$.
Mu'alem and Schapira~\cite{MualemS18}, showed a lower bound of
$2-1/n$, for both types of randomized truthful
mechanisms. Christodoulou, Koutsoupias and Kov{\'a}cs~\cite{CKK10}
extended the lower bound for fractional mechanisms, where each task
can be fractionally allocated to multiple machines. They also showed a
fractional mechanism with a guarantee of $(n+1)/2$.
Even for the special case of two machines the upper and lower bounds
are close but not tight~\cite{LuY08a,ChenDZ15}.

In the Bayesian setting, Daskalakis and Weinberg \cite{DaskalakisW15}
showed a mechanism that is at most a factor of 2 from the {\em optimal
  truthful mechanism}, but not with respect to optimal
makespan. Chawla et al.~\cite{ChawlaHMS13} provided bounds of
prior-independent mechanisms (where the input comes from a probability
distribution unknown to the mechanism). Giannakopoulos and
Kyropoulou~\cite{GiannakopoulosK17} showed that the VCG mechanism
achieves an approximation ratio of $O( \log n/\log \log n )$ under
some distributional and symmetry assumptions. Finally, in a recent
work, Christodoulou, Koutsoupias and Kov{\'a}cs~\cite{CKK21} showed
positive results for settings where each task can only be allocated to
at most two machines, a property shared by the instances of the lower
bound of this paper.

\section{Preliminaries}
\label{sec:preliminaries}

In the unrelated machines scheduling problem, we are given a set $M$ of $m$ tasks that need to be scheduled on a set $N$ of $n$ machines. The processing time or cost that each machine $i$ needs to process task $j$ is $t_{ij}$, and the completion time of machine $i$ for a subset $S$ of tasks is equal to the sum of the individual task costs $t_i(S)=\sum_{j\in S}t_{ij}$. The objective is to find an allocation of tasks to machines that minimize the \emph{makespan}, that is, the maximum completion time of a machine.
 
\subsection{Mechanism design setting}

We assume that each machine $i\in N$ is controlled by a selfish agent (player) that is reluctant to process tasks and the costs $t_{ij}$ is private information known only to them (also called the {\em type} of agent $i$).  The set $\mathcal{T}_i$ of possible types of agent $i$ consists of all vectors $b_i=(b_{i1},\ldots, b_{im})\in \mathbb{R}_+^{m}.$ Let also $\mathcal{T} = \times_{i\in N}\mathcal{T}_i$ be the space of type profiles.

A mechanism defines for each player $i$ a set $\mathcal{B}_i$ of available strategies the player can choose from. We consider \emph{direct revelation} mechanisms, i.e., $\mathcal{B}_i=\mathcal{T}_i$ for all $i,$ meaning that the players strategies are to simply report their types to the mechanism. Each player $i$ provides a \emph{bid} $b_i\in \mathcal{T}_i$, which not necessarily matches the true type $t_i$, if this serves their interests. A mechanism $(A,\mathcal P)$ consists of two parts:
\paragraph{An allocation algorithm:} The allocation algorithm $A$ allocates the tasks to machines based on the players' inputs (bid vector) $b=(b_1,\ldots ,b_n)$. Let $\mathcal{A}$ be the set of all possible partitions of $m$ tasks to $n$ machines. The allocation function $A:\mathcal{T}\rightarrow \mathcal{A}$ partitions the tasks into the $n$ machines; we denote by $A_i(b)$ the subset of tasks assigned to machine $i$ for bid vector $b=(b_1,\ldots,b_n)$.
\paragraph{A payment scheme:} The payment scheme $\mathcal P=(\mathcal P_1,\ldots,\mathcal P_n)$ determines the payments, which also depends on the bid vector $b$. The functions $\mathcal P_1,\ldots,\mathcal P_n$ stand for the payments that the mechanism hands to each agent, i.e., $\mathcal P_i:\mathcal{T}\rightarrow \R$.

\medskip

The {\em utility} $u_i$ of a player $i$ is the payment that they get minus the {\em actual} time that they need to process the set of tasks assigned to them, $u_i(b)=\mathcal P_i(b)-t_i(A_i(b))$. We are interested in \emph{truthful} mechanisms. A mechanism is truthful, if for every player, reporting his true type is a \emph{dominant strategy}. Formally,
$$u_i(t_i,b_{-i})\geq u_i(t'_i,b_{-i}),\qquad \forall i\in [n],\;\;
t_i,t'_i\in \mathcal{T}_i, \;\; b_{-i}\in \mathcal{T}_{-i},$$ where notation $\mathcal{T}_{-i}$ denotes all parts of $\mathcal{T}$ except its $i$-th part.

The quality of a mechanism for a given type $t$ is measured by the makespan $\mech(t)$ achieved by its allocation algorithm $A$, $ \mech(t) = \max_{i\in N}t_{i}(A_i(t))$, which is compared to the optimal makespan $\opt(t)=\min_{A\in \mathcal{A}}\max_{i\in N}t_i(A_i)$.

It is well known that only a subset of algorithms can be allocation algorithms of truthful mechanisms. In particular, no mechanism's allocation algorithm is optimal for every $t$, so it is natural to focus on the approximation ratio of the mechanism's allocation algorithm. A mechanism is \emph{$c$-approximate}, if its allocation algorithm is $c$-approximate, that is, if $c\geq\frac{\mech(t)}{\opt(t)}\;$ for all possible inputs $t$.

In this work, we do not place any requirements on the time to compute the mechanism's allocation $A$ and payments $\mathcal P.$ In other words, the lower bound does not make use of any computational assumptions.

\subsection{Weak monotonicity}

Mechanisms consists of two algorithms, the allocation algorithm and the payment algorithm. However, we are only interested in the performance of the allocation algorithm. Therefore it is natural to ask whether it is possible to characterize the class of allocation algorithms that are part of a truthful mechanism with no reference to the payment algorithm. Indeed the following definition provides such a characterization.

\begin{definition} \label{def:wmon} An allocation algorithm $A$ is called {\em weakly monotone (WMON)} if it satisfies the following property: for every two inputs $t=(t_i,t_{-i})$ and $t'=(t'_i,t_{-i})$, the associated allocations $A$ and $A'$ satisfy $$t_i(A_i)-t_i(A'_i)\leq t'_i(A_i)-t'_i(A'_i).$$
  An equivalent condition, using $a_{ij}(t)\in\{0,1\}$ indicator variables of whether task $j$ is allocated to player $i$, is
  \begin{align*}
    \sum_{j\in M} (a_{ij}(t')-a_{ij}(t))(t_{ij}'-t_{ij})\leq 0.
  \end{align*}
\end{definition}

It is well known that the allocation function of every truthful mechanism is weakly monotone~\cite{BCR+06}. Although it is not needed in establishing a lower bound on the approximation ratio, it turns out that this is also a sufficient condition for truthfulness in convex domains~\cite{SY05} (which is the case for the scheduling domain).

A useful tool in our proof relies on the following immediate consequence of weak monotonicity (see \cite{CKK20} for a simple proof). Intuitively, it states that when you fix the costs of all players for a subset of tasks, then the {\em restriction} of the allocation to the rest of the tasks is still weakly monotone.

\begin{lemma}\label{lem:restriction} Let $A$ be a weakly monotone allocation, and let
  $(S,T)$ a partition of $M$. When we fix the costs of the tasks of $T$, the restriction of the allocation $A$ on $S$ is also weakly monotone. %
\end{lemma}

The following implication of weak monotonicity, that was first used in \cite{NR01}, is a standard tool for showing lower bounds for truthful mechanisms (see for example \cite{NR01,ChrKouVid09,MualemS18,ADL09,giannakopoulos2020,DS20}).

\begin{lemma}\label{lemma:tool}
  Consider a truthful mechanism $(A,\mathcal P)$ and its allocation for a bid vector $t$. Let $S$ be a subset of the tasks allocated to player $i$ and let $S'$ be a subset of the tasks allocated to the other players. Consider any bid profile $t'=(t'_i,t_{-i})$ that is obtained from $t$ by decreasing the values of player $i$ in $S$, i.e., $t'_{ij}< t_{ij}, j\in S$, and increasing the values of player $i$ in $S'$, i.e., $t'_{ij}> t_{ij}, j\in S'$. Then the allocation of player $i$ for $t$ and $t'$ agree for all tasks in $S\cup S'$.
\end{lemma}

Notice that this lemma guarantees that only the set $S$ of tasks allocated to player $i$ remains the same. This does not preclude changing the allocation of the other players for the tasks in $S'$, unless there are only two players. This is a major obstacle in multi-player settings, that we avoid in this work by focusing on only two players for each task.

\subsection{Characterization of truthful mechanisms for 2 players and 2 tasks.}
\label{sec:char-all-truthf}

In this section we discuss the class of truthful $2\times 2$ mechanisms (i.e.,
for 2 machines and 2 tasks). We present here a full characterization of these
mechanisms as it appears in the full version of~\cite{CKK20}.

There are other such characterizations in the literature (for example~\cite{DS08,ChristodoulouKV08}), but they are somewhat incomplete, in the sense that they either impose additional restrictions to mechanisms (for example, bounded approximation ratio) or consider a bigger domain with negative values. These characterizations do not fit the purpose of this work. For example, we deal with domains with multiple players and tasks and conditions on the approximation ratio of restrictions to only 2 tasks are not applicable.

An important aspect of our approach is to restrict the allocation of each task to only two specific machines, which we call player (machine) $t$ and player $s$. To achieve this, we fix two large constants $\A$ and $B$ with $\A\gg B\gg 1$, and consider instances in which the value for the $t$-player can be arbitrary, for the $s$ player in $[0,B)$, and for the other players equal to $\A$. Since the values of the $s$-player are much smaller than $\A$, if the mechanism allocates even one task to a player with value $\A$, its approximation ratio is really high.

A convenient aspect of the characterization of~\cite{CKK20}, is that it characterizes the $2\times 2$ mechanisms with values exactly in this domain. We first present the theorem and then define and discuss the different types of mechanisms. The characterization assumes that for sufficiently high values of $t$-player, both tasks will be allocated to the $s$-player, otherwise the approximation ratio is unbounded, even when there are more players and tasks.

\begin{theorem}[Characterization of $2\times 2$
  mechanisms~\cite{CKK20}] \label{theo:addchar} For every
  $B\in \R_{>0}\cup\{\infty\}$, every weakly monotone allocation for two tasks
  and two additive players with bids $t\in [0,\infty)\times [0,\infty)$ and
  $s\in[0,B)\times [0,B)$ ---such that for every $s$ there exists $t$ for which
  both tasks are allocated to the second player--- belongs to the following
  classes: (1) relaxed affine minimizers (including the special case of affine
  minimizers), (2) relaxed task independent mechanisms (including the special
  case of task independent mechanisms), (3) 1-dimensional mechanisms, (4) constant
  mechanisms.\footnote{To be precise, the restriction of every weakly monotone
    allocation to strictly positive $t$ and $s$ is of type (1) to (4); on the
    border $0$ there might be singularities, see \cite{CKK20} for details. For
    this reason, in the main argument we use strictly positive $t$ and $s$.}
\end{theorem}

Note that the theorem applies also to unbounded domain, i.e., the case of
$B=\infty$.  Some mechanisms can be of two types; for example the VCG mechanism
is both an affine minimizer and task independent.

To discuss the different types of mechanisms of the theorem, let $(A,\mathcal P)$ be a truthful $2\times 2$ mechanism, where $A$ is a weakly monotone allocation function, and $\mathcal P$ denotes the payment function. For given fixed $s\in[0,B)\times [0,B)$ let us consider the allocation for the $t$-player depending on his own bids $(t_1,t_2)$ (see Figure~\ref{fig:shapesPure}). The allocation \emph{regions} $R_{12}$ (resp.  $R_1, R_2, R_\emptyset$) $\,\subseteq \mathbb R_{\geq 0}^2$ are defined to be the \emph{interior} (wrt. $\mathbb R_{\geq 0}^2$) of the set of all $t$ values such that the $t$-player gets the set of tasks $\{1,2\}$ (resp. $\{1\},\,\{2\},\,\emptyset$). The allocation of the points on the boundaries can take any allocation of the adjacent regions.

\begin{figure}
  \centering
  \begin{tikzpicture}[scale=0.64]

    \draw[->] (0,0) -- (6,0) node[anchor=north] {$t_{1}$};
    \draw[->] (0,0) -- (0,6) node[anchor=east] {$t_{2}$};
    \draw[very thick, blue] (1.9,6) -- (1.9,3.68) -- (4,1.58) -- (6,1.58) ; \draw[very thick, blue] (0, 3.68) -- (1.9,3.68) ;

    \draw[ultra thick, red, dashed] (2,6) node[anchor = south] {$\psi_1(t_2,s)$} -- (2,3.7) -- (4.1,1.6) ; \draw[ultra thick, red, dashed] (4.1,1.6) -- (4.1,0);

    \draw[very thick, blue] (4,0) -- (4, 1.58) -- (1.9, 3.68);

    \draw (1.5,5) node[anchor=east] {$R_1$}; \draw (1.7,1) node[anchor=east] {$R_{12}$}; \draw (5.5,1) node[anchor=east] {$R_2$}; \draw (5.5,5) node[anchor=east] {$R_{\emptyset}$};

    \draw (3,-0.5) node[anchor=north] {(a)};
  \end{tikzpicture}
  \begin{tikzpicture}[scale=0.64]

    \draw[->] (0,0) -- (6,0) node[anchor=north] {$t_{1}$};
    \draw[->] (0,0) -- (0,6) node[anchor=east] {$t_{2}$};
    \draw[very thick, blue] (4.1,6) -- (4.1, 3.68) -- (1.9,1.58) -- (0,1.58) ;

    \draw[ultra thick, red, dashed] (4.2, 6) node[anchor = south] {$\psi_1(t_2,s)$} -- (4.2, 3.7) -- (2,1.6) ; \draw[ultra thick, red, dashed] (2,1.6) -- (2,0);

    \draw[very thick, blue] (1.9,0) -- (1.9, 1.58); \draw[very thick, blue] (4.1,3.68) -- (6, 3.68) ;

    \draw (1.5,5) node[anchor=east] {$R_1$}; \draw (1.7,1) node[anchor=east] {$R_{12}$}; \draw (5.5,1) node[anchor=east] {$R_2$}; \draw (5.5,5) node[anchor=east] {$R_{\emptyset}$};

    \draw (3,-0.5) node[anchor=north] {(b)};
  \end{tikzpicture}
  \begin{tikzpicture}[scale=0.64]
    \draw[->] (0,0) -- (6,0) node[anchor=north] {$t_{1}$};
    \draw[->] (0,0) -- (0,6) node[anchor=east] {$t_{2}$};
    \draw[very thick, blue] (3,0) -- (3,6) ; \draw[very thick, blue] (0,3) -- (6,3) ;

    \draw[ultra thick, red, dashed] (3.1, 0) -- (3.1,6) node[anchor = south] {$\psi_1(t_2,s)$} ;

    \draw (1.5,5) node[anchor=east] {$R_1$}; \draw (1.7,1) node[anchor=east] {$R_{12}$}; \draw (5.5,1) node[anchor=east] {$R_2$}; \draw (5.5,5) node[anchor=east] {$R_{\emptyset}$};

    \draw (3,-0.5) node[anchor=north] {(c)};
  \end{tikzpicture}

  \caption{\small The allocation to the $t$-player depending on their own bid vector $(t_1,t_2)$ for fixed values $s=(s_1,s_2)$ of the other player: (a) quasi-bundling allocation; (b) quasi-flipping allocation; (c) crossing allocation.  The boundary $\psi_1(t_2,s)$ for task $1$ is shown by a dashed line.}
  \label{fig:shapesPure}
\end{figure}
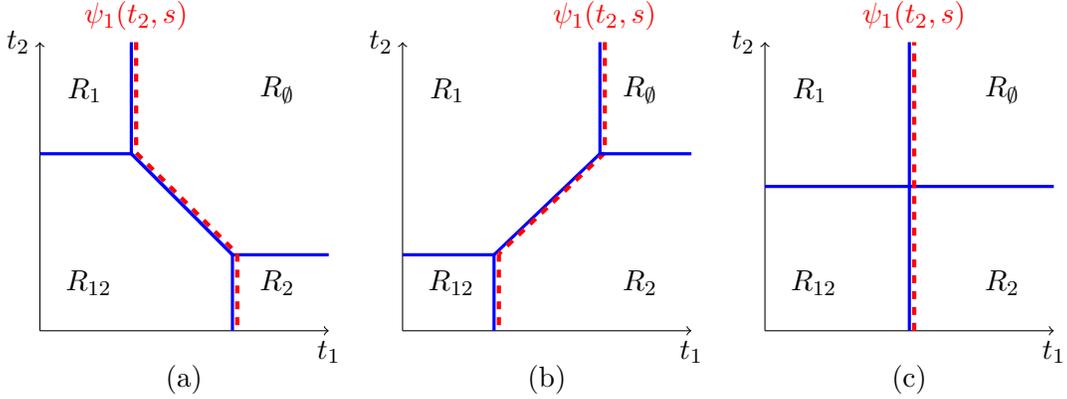

It is known that in the case of two tasks, for a fixed $s$, these regions in a
weakly monotone allocation subdivide $\mathbb R_{\geq 0}^2$ basically in three
possible ways (as in Figure~\ref{fig:shapesPure}), which are characteristic for
the type of the whole allocation-function $A$. Similar subdivisions for multiple
tasks exist in higher dimensions~\cite{ChristodoulouKV08,
  Vid09}. %

\begin{definition} \label{def:boundary} The allocation of a mechanism is defined
  by its boundary functions $\psi_i(t_{-i},s)$, the infimum of the values of
  task $t_i$ for which the $t$-player does not get task $i$. In other words,
  the $t$-player gets task $i$ when $t_i<\psi_i(t_{-i},s)$ and it does not get
  it when $t_i>\psi_i(t_{-i},s)$. 
\end{definition}

For a given $s,$ we call the allocation for the $t$-player 
  \begin{itemize}
  \item \emph{quasi-bundling}, if there are at least two points $t\neq t'$ on
    the boundary of $R_{12}$ and $R_{\emptyset}$. In this case, we say that the
    boundary between $R_{12}$ and $R_{\emptyset}$ is a \emph{bundling boundary}.
  \item \emph{quasi-flipping}, if there are at least two points $t\neq t'$ on
    the boundary of $R_{1}$ and $R_{2}$.  
  \item \emph{crossing}, otherwise.  (See Figure~\ref{fig:shapesPure} for an
    illustration.)
  \end{itemize}
 
\paragraph{Relaxed affine minimizers.}
An allocation is an \emph{affine minimizer}, if there exist positive constants $\lambda'$, $\lambda$ and constants $\pi_a\in \mathbb{R}\cup\{-\infty,\infty\}$ for each allocation $a\in\{12,1,2,\emptyset\}$, so that for every input $(t,s)$ the allocation of the $t$-player minimizes the corresponding expression among the following:
\begin{align} \label{eq:aff-min} \lambda'(t_1+t_2)+\pi_{12},\qquad \lambda'\, t_1+\lambda\, s_2+\pi_{1},\qquad \lambda'\, t_2+\lambda\, s_1+\pi_{2},\qquad \lambda(s_1+s_2)+\pi_{\emptyset}.
\end{align}
An affine minimizer allocation has typically either the form of Figure~\ref{fig:shapesPure}.a or~\ref{fig:shapesPure}.b.

Since we consider nonnegative values, the domain is bounded below. When an
affine minimizer has the bundling form (Figure~\ref{fig:shapesPure}.a), it may
be the case that for small $s_1+s_2$ and $t_1+t_2$, the mechanism locally `looks
like' a bundling mechanism, with only the regions $R_{\emptyset}$ and $R_{12}$
for both players. If this is the case for \emph{all} $(s_1,s_2)$ with
$s_1+s_2<D_s$ and \emph{all} $(t_1,t_2)$ such that $t_1+t_2<D_t$ for some given
$D_t,D_s>0,$ then the mechanism becomes locally a bundling mechanism. As such,
the boundary functions for the allocation need not remain linear for these small
$s$ and $t$ and the boundary between $R_{12}$ and $R_{\emptyset}$ is given by
$t_1+t_2=\zeta(s_1+s_2)$, for some non-decreasing function $\zeta$. See
Figure~\ref{fig:ex-relaxed-minimizer} for an example. For a formal definition of
these mechanisms, which are called \emph{relaxed affine minimizers}, we refer
the reader to~\cite{CKK20}, where they were introduced.

We will refer to the (potentially) non-linear bundling allocation of a relaxed
affine minimizer as its ``bundling tail''. In other words, the bundling tail of
a relaxed affine minimizer is determined by the values of $s$ for which the
allocation of the $t$-player contains only the regions $R_{12}$ and
$R_{\emptyset}$. Note that affine minimizers is simply the special subclass of
these mechanisms in which the bundling tail is missing or conforms
to~\eqref{eq:aff-min}.

\emph{The important aspect of affine minimizers and of relaxed affine minimizers
outside their bundling tail is that the boundary $\psi_i(t_{-i}, s)$ is a
truncated linear function in $s_i$, and in particular of the form
\begin{align*}
  \psi_i(t_{-i}, s)&=\max(0, \, \lambda s_i - \gamma(t_{-i},s_{-i})),
\end{align*}
where $\lambda$ is a constant and $\gamma(t_{-i},s_{-i})$ a function that does
not depend on the values of task $i$.} Note that when we consider a $2\times 2$
mechanism as a restriction of a mechanism with multiple players and tasks, the
coefficient $\lambda$ is not an absolute constant but it can depend in an
arbitrary way on values outside the $2\times 2$ slice (see
Definition~\ref{def:slice} for the definition of slice); similarly for the
function $\gamma$.

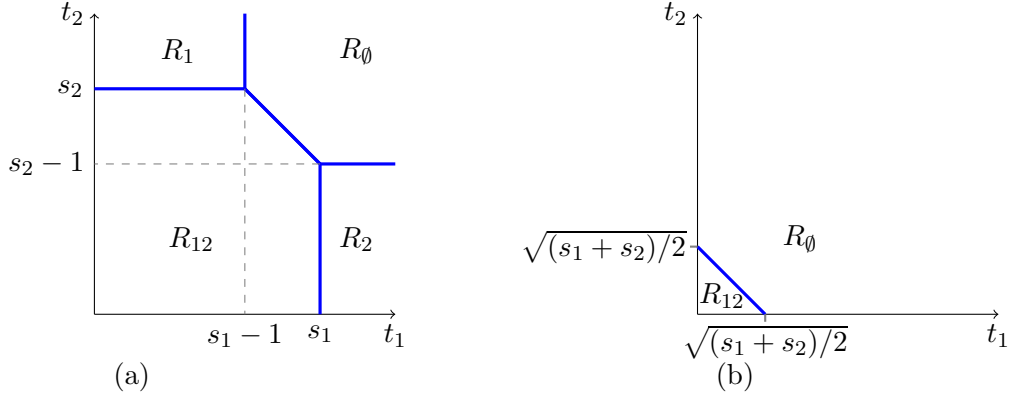
\begin{figure}
  \centering
  \begin{tikzpicture}[scale=0.5]

    \draw[->] (0,0) -- (8,0) node[anchor=north] {$t_1$};
    \draw[->] (0,0) -- (0,8) node[anchor=east] {$t_2$};

    \draw[very thick, blue] (0, 6) node[anchor=east,black] {$s_2$} -- (4, 6) -- (6, 4) -- (8, 4); \draw[very thick, blue ] (6, 0) node[anchor=north,black] {$s_1$}-- (6, 4) -- (4, 6) -- (4, 8);

    \draw[dashed, gray] (4,6) -- (4,0) node[anchor=north,black] {$s_1-1$}; \draw[dashed, gray] (6,4) -- (0,4) node[anchor=east,black] {$s_2-1$};

\draw (1.5,7) node[anchor=west] {$R_1$}; \draw (1.7,2) node[anchor=west] {$R_{12}$}; \draw (7.7,2) node[anchor=east] {$R_2$}; \draw (7.7,7) node[anchor=east] {$R_{\emptyset}$};

    \draw (1,-1) node[anchor=north] {(a)};
  \end{tikzpicture}
\hspace*{1cm}
  \begin{tikzpicture}[scale=0.5]

    \draw[->] (0,0) -- (8,0) node[anchor=north] {$t_1$};
    \draw[->] (0,0) -- (0,8) node[anchor=east] {$t_2$};

    \draw[very thick, blue] (1.8,0) node[anchor=north,black]
    {$\sqrt{(s_1+s_2)/2}$} -- (0,1.8) node[anchor=east,black]
    {$\sqrt{(s_1+s_2)/2}$};
    \draw[thick, gray] (1.8,0) -- (1.8,-0.2);
    \draw[thick, gray] (0,1.8) -- (-0.2,1.8);

    \draw (-0.23,0.5) node[anchor=west] {$R_{12}$};
    \draw (2,2) node[anchor=west] {$R_{\emptyset}$};

    \draw (1,-1) node[anchor=north] {(b)};
  \end{tikzpicture}
  
  \caption{\small An example of a relaxed affine minimizer, which shows the
    allocation of the $t$-player for two distinct values of $s$. The left figure
    shows the allocation for some $s_1\geq 1$ and $s_2\geq 1$; in this case, the
    boundaries are linear; for example, when $t_2\in [s_2-1,s_2]$ the allocation
    for task 1 is determined by comparing $t_1$ to
    $\psi_1(t_2,s)=s_1+s_2-1$. The right figure shows the allocation at the
    bundling tail for some $s_1,s_2\leq 1$, in which case the boundary does not
    have to be linear; for example for $t_2\in[0,\sqrt{(s_1+s_2)/2}]$, we have
    $\psi_1(t_2,s)=\sqrt{(s_1+s_2)/2}$.}
  \label{fig:ex-relaxed-minimizer}
\end{figure}

\paragraph{Relaxed task independent allocations.}
An allocation function $A$ is \emph{task independent} if for both tasks the
allocation of task $i$ depends only on the input values $s_i$ and $t_i.$ For the
$t$-player, the boundary of $t_1,$ i.e., the lowest value, above which $t_1$
does not get task $1,$ is determined by an \emph{arbitrary non-decreasing
  function $\psi_1:[0,B)\rightarrow [0,\infty)$ of $s_1$} and analogously for
the boundary $\psi_2(s_2).$ %
Geometrically, in a task independent mechanism, the allocations of both players are always crossing (Figure~\ref{fig:shapesPure}.c).  In a \emph{relaxed task independent} mechanism the latter property is fulfilled in all but countably many $s$ (resp. $t$) points, in which both $\psi_1$ and $\psi_2$ (resp. $\psi_1^{-1}$ and $\psi_2^{-1}$) have a jump discontinuity. See~\cite{CKK20} for a formal definition. \emph{The important property that we use here is that every relaxed task independent allocation is identical with a task independent allocation on $t_i\in [0,\infty)\setminus T_i, \,s_i\in [0,B)\setminus S_i, $ where the $T_1, T_2, S_1, S_2$ are countable sets.}

\paragraph{1-dimensional mechanisms.}

In a \emph{1-dimensional mechanism} at most two possible allocations are ever realized. %
If the two occuring allocations are $\emptyset$ and $12,$ we call the mechanism \emph{bundling mechanism}. One can consider bundling 1-dimensional mechanisms as degenerate \emph{relaxed affine minimizers} with $\pi_1=\pi_2=\infty$.

The other cases when the allocations to the $t$-player are $\emptyset$ and $1$ (or $\emptyset$ and $2$) are degenerate \emph{task independent} allocations, and for our purposes they can be treated as task independent. %

\paragraph{Constant mechanisms.}
In a \emph{constant or dictatorial mechanism} the allocation is independent of the bids of at least one of the players. This property can also be interpreted as being an affine minimizer with a multiplicative constant $\lambda=0.$ %

\section{Lower Bound}
\label{sec:lb}

In this section, we give a proof of our main result. In
Sections~\ref{sec:construction} and~\ref{sec:def-notation}, we
describe the general setting, provide important definitions, and set
the main goal. With these at hand, we outline the proof
in Section~\ref{sec:outline}. Then in
Sections~\ref{sec:general-lemmas},~\ref{sec:slice}, and~\ref{sec:sum-up} we provide the proofs
of all the technical lemmas needed in order to establish the main
theorem, the proof of which is presented in the last subsection (Section~\ref{sec:appr-ratio-sqrtn}).

\subsection{The construction}
\label{sec:construction}

We consider instances, with $n$ players and $m=(\ell+1)(n-1)+n$ , for
some large $\ell$ to be determined later. Player $0$ is special and
for convenience we use the symbol $t$ for its values; sometimes we
refer to it as the $t$-player or $0$-player. We use the symbol $s$ for
the values of the remaining players $1,\ldots,n-1$, and sometimes we
refer to them as the $s$-players.

There are $n$ tasks $d_0,\ldots,d_{n-1}$, which are special that are
called \emph{dummy tasks}; they play a limited role in the proof and
their only purpose is to increase the lower bound from $\sqrt{n-1}$ to
$1+\sqrt{n-1}$.

The remaining tasks are partitioned into $n-1$ clusters
$C_1,\ldots, C_{n-1}$, where each cluster $C_i$ contains $\ell+1$
tasks and is associated with player $i\in
[n-1]$. %
We call two tasks $j$ and $j'$ that belong to the same cluster \emph{siblings}.

\begin{definition}[Range of input values] \label{def:range-of-values}
  The processing times for a task $j\in C_i$, $i\in [n-1]$, is
  described by two values $t_j$ and $s_j$, as follows:
  \begin{itemize}
  \item player $0$ has processing time $t_j\in (0,\infty)$; with very few exceptions, the argument uses values $t_j\in(0,1]$.
  \item player $i$ has processing time $s_j\in (0,B)$, for some $B$; again with very few exceptions, the argument uses values $s_j\in(0,1]$.
  \item every other player $k\not\in \{0,i\}$ has processing time $\A$, for some
    $\A$.
  \end{itemize}
  Dummy task $d_i$ has value $\A$ for all players except for player $i$ for
  which the value is initially 0.

  The values of $B$ and $\A$ are arbitrarily large functions of $n$ (exponential
  functions suffice to get a good approximation ratio) with
  $\A/((\ell+1) B)\gg n$.
\end{definition}

Since every task has at least one processing time in $[0,B)$, no algorithm with approximation ratio less than $\A/((\ell+1) B)\gg n$, allocates any task to players with value $\A$.

An instance (input) $T$ is described only by two values $t_j$ and $s_j$ per task $j$. Let's denote by $t$, and $s$ the respective vectors, hence, $T=(t,s)=(t_j, s_j)_{j\in [m]}$.

\subsection{Definitions}
\label{sec:def-notation}

We use the following fixed values throughout this section.
\begin{itemize}
\item $\alpha=1/\sqrt{n-1}$
\item $\beta$, an arbitrarily small positive value%
\item $\delta$, a small value; think of this as $n^{-2}$; any value $o(n^{-3/2})$ gives lower bound $1+(1-o(1))\sqrt{n-1}$. For simplicity, we will assume that $\delta$ is selected so that \emph{$2n/\delta$ is an integer}.
\item $\delta'=2\delta$, upper bound on the cost of a trivial cluster (see definition of a trivial cluster below)%
\item $\rho,$ the targeted lower bound on the approximation ratio,
  which is given
  by $$\rho=1-\delta'+\min\left\{\frac{1}{\alpha+(n-1)\delta'},
    \frac{(n-1)\alpha}{1+(n-1)\delta'}\right\}=1+(1-o(1))\sqrt{n-1},$$
\item $\ell+1$ number of tasks per cluster; exponential in $n$, and greater than $1+3n^3 \left(\frac{3n}{\delta}\right)^{n-3}=(n/\delta)^{\Theta(n)}$.
\end{itemize}

The proof is based on the characterization of $2\times 2$
mechanisms. To be able to use it, we fix all other values except of
the values of two tasks and then use the characterization. The next
definition formalizes this.

\begin{definition}[Slice] \label{def:slice}
  Fix an instance $T$ and two tasks $p$ and $p'$. The set of instances
  that agree with $T$ on all tasks except on the tasks $p$ and $p'$ is
  called a $(p,p')$-slice for $T$ or simply slice of tasks $p$ and
  $p'$. A slice may involve values of 3 different players, if $p,p'$
  belong to different clusters, or values of 2 players if $p,p'$ are
  siblings. In the latter case, the allocation of these tasks (due to
  Lemma~\ref{lem:restriction}) is defined by a $2\times 2$ mechanism
  which we call $(p,p')$-slice mechanism for $T$.
\end{definition}

To take advantage of weak monotonicity, we need to consider
perturbations of instances, that is, instances that differ by a small
amount from a given one. Here is a precise definition of the
perturbations that we use.

\begin{definition}[Perturbations of an instance]
  Fix an instance $T=(t,s)$ and a set of tasks $P$. Let $V$ be a vector of open intervals $V_j$, one for each task $j\in P$, such that
  \begin{itemize}
  \item $V_j=(t_j,t_j+ \theta_j)$, for every task $j\in P$, for some $\theta_j\in (0,\beta)$
  \end{itemize}
  The set of instances $T'=(t',s)$ with $t'_j\in V_j$, when $j\in P$, and $t'_j=t_j$ when $j\not\in P$, is called a set of \emph{$V$-perturbations of $T$ for tasks $P$} or simply set of \emph{$V$-perturbations} of $T$, when $P$ is understood from the context.  %
\end{definition}

Note that we consider perturbations only of $t$-values and only in one
direction (towards higher values). The values of instance $T$
itself are \emph{not} in the perturbation for tasks in $P$.

Perturbations of an instance $T$ satisfy a few important properties:
\begin{itemize}
\item for any given allocation, the cost of all instances of a perturbation is almost the same as the cost of $T$ (within $(n-1)|P|$).
\item they allow us to select points not on boundaries of the mechanism, thus when we apply weak monotonicity for the $0$-player, we can guarantee that certain allocations do not change.
\end{itemize}

The central part of the argument is an induction on the number $k$ of clusters. The values of the tasks in the remaining $n-k-1$ clusters, which we call trivial clusters, play a limited role, but it is important that they do not affect substantially the approximation ratio. We allow their values to be arbitrary (within the limits of Definition~\ref{def:range-of-values}), but we require that one of the values is very small in the following sense.

\begin{definition}[Trivial cluster]
  A cluster is called \emph{trivial} for a given instance $T$ if the optimal allocation for all tasks of the cluster has cost at most $\delta'$.
\end{definition}

We usually select a single task from each non-trivial cluster. We now fix the terminology for such sets of tasks:

\begin{definition}[Regular set of tasks]
  A set of tasks is called \emph{regular} if they are from different clusters. %
\end{definition}

We make repeated use of the following set of instances:

\begin{definition}[Standard instance and $\hat T(P)$]
  An instance $T$ is \emph{standard} for a set of clusters $\cal C$ if the following conditions hold
  \begin{itemize}
  \item the value of every task $j\in \cup \cal C$ is $[t_j=\beta, \, s_j=1]$, and
  \item the remaining clusters are trivial
  \end{itemize}
  We also say that \emph{$T$ is standard for a regular set of tasks $P=\{p_1,\ldots,p_k\}$}, if it is standard for the set of clusters that intersect $P$. We denote by $\hat T(P)$ the instance that agrees with $T$ everywhere except of tasks in $P$ for which $[\hat t_{p_i}=\alpha, \hat s_{p_i}=1]_{i=1}^k$.  See Figure~\ref{fig:Good Set} for an illustration.
\end{definition}

The following definition is at the heart of the proof. Roughly speaking, the aim of the proof is to show by induction that there exist good sets of $n-1$ tasks, otherwise the mechanism has high approximation ratio.

\begin{definition}[Good set of tasks] \label{def:goodness} Fix a
  mechanism, a set of regular tasks $P=\{p_1,\ldots,p_k\}$ from a set
  of clusters $\cal C$, and a standard instance $T$ for $P$. The set
  of tasks $P$ is called \emph{good for instance $T$} if there exists
  a vector $V$ of open intervals for the tasks in $P$ such that the
  mechanism allocates all tasks in $P$ to the 0-player {\em for every}
  instance in the set of $V$-perturbations of $\hat T(P)$ for tasks in $P$. We call $V$
  the witness of goodness of $P$. (See Figure~\ref{fig:Good Set} for
  an illustration.)

  If no such $V$ exists, we call $P$ a \emph{bad} set. A singleton bad set will
  be simply called \emph{bad task}. For technical reasons, we will also call a
  task bad when it has all the above properties of a bad task, but its $s$-value
  is $1-\epsilon$, for some arbitrarily small $\epsilon>0$.\footnote{In
    particular, a task will be bad if for every $\epsilon'>0,$ there is a
    $0\leq \epsilon< \epsilon',$ so that setting its $s$-value to $1-\epsilon,$
    the task has the above properties. This will serve to exclude singular
    points of relaxed task independent mechanisms.}
\end{definition}

\begin{figure}[h]
  \begin{align*}
    T=  \left[
    \begin{array}{c c c c c c c c c}
      \beta & \beta & \beta & \beta & \beta & \beta & \delta & \beta & \beta \\
      1 & 1 & 1 &  &  &  &  &  &  \\
        &   &   & 1 & 1 & 1 &   &   &  \\
        &   &   &   &   &   & 1/2 & 1 & 1 \\
    \end{array}\right],\qquad
    \hat T(\{1,5\})=  \left[
    \begin{array}{c c c c c c c c c}
      \alpha & \beta & \beta & \beta & \alpha & \beta & \delta & \beta & \beta \\
      1 & 1 & 1 &  &  &  &  &  &  \\
             &   &   & 1 & 1 & 1 &   &   &  \\
             &   &   &   &   &   & 1/2 & 1 & 1 \\
    \end{array}\right]
  \end{align*}

  \caption{\small $T$ and $\hat T$ are two instances with $4$ players and $3$ clusters, where
    rows correspond to players and columns to tasks. The values $\A$ and the dummy tasks are omitted. Typical values of the constants are $\beta\approx 0$, $\delta\approx n^{-2}$, $\alpha\approx n^{-1/2}$. $T$
    is a standard instance for clusters $\{C_1,C_2\}$. $C_3$ is a trivial cluster because it
    has cost $\delta+2\beta\leq \delta'$.  %
    If there exists a $V$ such that tasks $1\in C_1,$ and $5\in C_2$ are assigned to the $0$-player for every instance in the set of $V$-perturbations of $\hat T(\{1,5\})$, then $\{1,5\}$ is a good set for $T$.}
  \label{fig:Good Set}
\end{figure}

Now that we have the definition of a set of good tasks, we can state the main
lemma of this section.

\begin{lemma}[Main Lemma] \label{lemma:main} At least one of the following three
  properties hold for every truthful mechanism:
  \begin{itemize}
  \item[(i)] the approximation ratio is at least $\rho$
  \item[(ii)] there exists a bad task
  \item[(iii)] there exists a good set of $n-1$ tasks.
  \end{itemize}
\end{lemma}

The aim of this section is to prove this lemma by assuming that Properties
\emph{(i)} and \emph{(ii)} do not hold and then show Property \emph{(iii)}.

\begin{assumption}\label{ass1}
  For the rest of this section, we assume that the approximation ratio is less than
  \begin{align*}
\rho=1-\delta'+\min\left\{\frac{1}{(\alpha+(n-1)\delta')}, \frac{(n-1)\alpha}{(1+(n-1)\delta')}\right\}
  \end{align*}
  and that there is no bad task.
\end{assumption}

Before we give the proof of the Main Lemma, we discuss how it will be used to
prove a lower bound of $1+\sqrt{n-1}$ on the approximation ratio (Section~\ref{sec:appr-ratio-sqrtn} contains a detailed derivation). We
provide here a high-level argument.  Property \emph{(i)} immediately implies the
lower bound. The existence of a bad task (Property \emph{(ii)}) means that there is
only one non-trivial task $j$ with values $[t_j=\alpha, \, s_j=1]$ and the
mechanism gives it to the $s$-player. In this case, the approximation ratio is
approximately $1/\alpha=\sqrt{n-1}$, which can be improved to $1+\sqrt{n-1}$
using the dummy tasks. Finally, a good set of $n-1$ tasks (Property \emph{(iii)}) has
approximation ratio approximately $(n-1)\alpha=\sqrt{n-1}$, which again can be
improved to $1+\sqrt{n-1}$ using the dummy tasks, giving the desired result.

To obtain a proof of the Main Lemma, we will show that there exists a good set
of $k$ tasks for every $k\in[n-1]$, by induction on $k$. To show existence of
good sets of $k$ tasks, we start with some set of $k$ tasks, which we call
potentially-good set, such that all its subsets of $k-1$ tasks are good. To
satisfy all the requirements in the proof, the precise structure of a
potentially-good set of tasks is complicated and it is detailed in the following
definition (see Figure~\ref{fig:good-potentially-good}).

\begin{definition}[Potentially-good set of tasks] \label{def:potentially-good} Fix a mechanism. A set of regular tasks $P=(p_1,\ldots,p_k)$ from a set of clusters $\cal C$ is called \emph{potentially-good for an instance $T$} if $T$ is standard for $P$ and the following conditions hold
  \begin{itemize}
  \item for every $i\in [k]$, $P_{-i}=P\setminus \{p_i\}$ is good for $T$; let $V^i$ denote a witness of goodness of $P_{-i}$
  \item $P_{-k}$ is good (with witness $V^k$) for every instance that results from $T$ when we replace the values\footnote{Note that the cluster of task $p_k$ must be trivial after the change. This is the point where we need the extra parameter $\delta'=2\delta$. Actually any value of $\delta'\geq\delta+(\ell+1)\beta$ works.
    } of task $p_k$ with $[t_{p_k}=\delta, \, s_{p_k}=q \delta/(2n)]$, for $q=0,1,\ldots,2n/\delta.$
  \end{itemize}
  The witness $V$ of potential goodness of $P$ is defined as follows: for task $p_i$ take the intersection of all relevant intervals in $V^j$, $j\neq i$.
\end{definition}

\begin{remark}
  One can replace $\delta/(2n)$ by $\delta/(2\rho)$ in the above
  definition to reduce the required number of $\ell+1$ tasks per
  cluster, but we opt for simplicity.
 \end{remark}
 \begin{figure}
  \centering
  \begin{tikzpicture}[scale=0.70]

    \draw[->] (0,0) -- (6,0) node[anchor=north] {$t_{p_k}$};
    \draw[->] (0,0) -- (0,6) node[anchor=east] {$t_{p_i}$};
    \draw[very thick, blue] (3.1,6) -- (3.1,3.68) -- (4,2.78) -- (6,2.78) ; \draw[very thick, blue] (0, 3.68) -- (3.1,3.68) ;
    \draw (3.0,0) node[anchor=north] {$\alpha$};
    \draw (0,3.0) node[anchor=east] {$\mathbf{\alpha}$};
    \draw[very thick, blue] (4,0) -- (4, 2.78) -- (3.1, 3.68);

    \draw (1.5,5) node[anchor=east] {$R_k$}; \draw (1.5,1) node[anchor=east] {$R_{ki}$}; \draw (5.5,1) node[anchor=east] {$R_i$}; \draw (5.5,5) node[anchor=east] {$R_{\emptyset}$};
    \fill[red] (3.0,3.0) circle (3pt);
    \draw[ultra thin, dashed] (0,3.0) -- (3.0,3.0) -- (3.0,0);
    \draw (3,-1.0) node[anchor=north] {(a)};
  \end{tikzpicture}
\hspace{3.0cm}%
\begin{tikzpicture}[scale=0.70]

    \draw[->] (0,0) -- (6,0) node[anchor=north] {$t_{p_k}$};
    \draw[->] (0,0) -- (0,6) node[anchor=east] {$t_{p_i}$};
    \draw[very thick, blue] (1.9,6) -- (1.9,3.68) -- (4,1.58) -- (6,1.58) ; \draw[very thick, blue] (0, 3.68) -- (1.9,3.68) ;
    \draw (3.0,0) node[anchor=north] {$\alpha$};
    \draw (0,3.00) node[anchor=east] {$\alpha$};
    \draw[very thick, blue] (4,0) -- (4, 1.58) -- (1.9, 3.68);

    \draw (1.5,5) node[anchor=east] {$R_k$}; \draw (1.5,1) node[anchor=east] {$R_{ki}$}; \draw (5.5,1) node[anchor=east] {$R_i$}; \draw (5.5,5) node[anchor=east] {$R_{\emptyset}$};
    \fill[red] (3.00,3.00) circle (3pt);
    \draw[ultra thin, dashed] (0,3.0) -- (3.0,3.0) -- (3.0,0);
    \draw (3,-1.0) node[anchor=north] {(b)};
  \end{tikzpicture}
  
  \caption{\small Allocation of the $t$-player in the $(p_i,p_k)$-slice for
    a regular set of tasks $P=\{p_i,p_k\},$ and some given standard instance $T$
    for $P$. Point $(\alpha,\alpha)$ represents $\hat T(P)$ when $P$ is good
    (left figure) and when $P$ is potentially-good but not good (right figure).}
  \label{fig:good-potentially-good}
\end{figure}
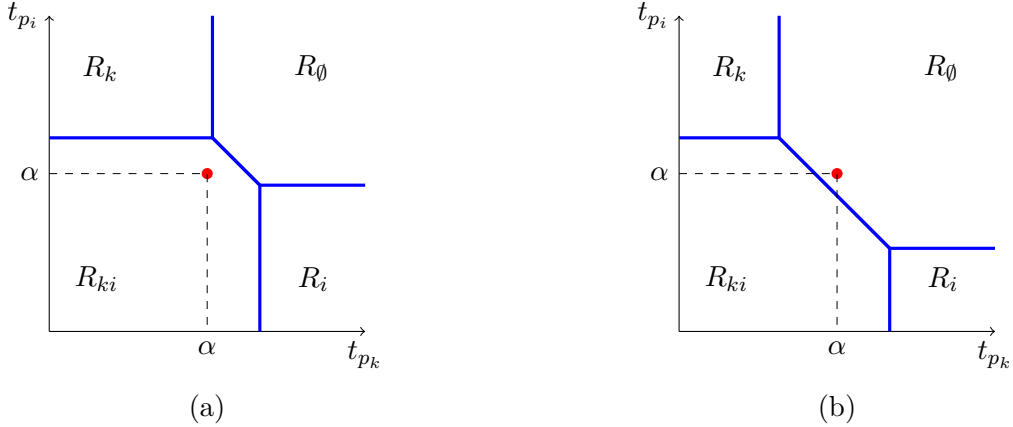

\subsection{Outline}
\label{sec:outline}

We now give a rough outline of the argument that establishes the Main Lemma
(Lemma~\ref{lemma:main}). We consider standard instances $T$ for sets of $k$
clusters $\cal C$. We will show that for $k=n-1$, there is a good set of
tasks. By induction on $k$, we show the stronger claim that there are \emph{many
  sets of tasks that are good}. The base case ($k=1$) is the assumption that all
single tasks are good (Assumption~\ref{ass1}).

We use a probabilistic argument to try to keep things simple. In particular, let
$1-b_k$ be (a lower bound on) the probability that %
a random regular set of tasks\footnote{In the argument, we select random tasks
  and instances. All random selections are from uniform distributions
  and independent.} $P=(p_1,\ldots,p_k)$ from $\cal C$ is good for
$T$. We want to show that the probability $b_k$ of $P$ being a bad set
of tasks is small, and in particular that $b_{n-1}<1$, which
establishes the existence of a good set of $n-1$ tasks.

\paragraph{Showing that $b_k$ is small.}
We show that $b_k$ is small by considering potentially-good sets
$P=(p_1,\ldots,p_k)$ of size $k$ that contain tasks from distinct
clusters $(Q_1,\ldots,Q_k)$ to establish the following two facts:
\begin{description}
\item[{\bf Fact a.}] the probability that a randomly selected set
  $P=(p_1,\ldots,p_k)$ is potentially-good is at least $1-(3n/\delta -
  1)b_{k-1}$ (Lemma~\ref{lemma:potentially-good}).
\item[{\bf Fact b.}] if $P=(p_1,\ldots,p_k)$ is a potentially-good set
  of tasks, either $P$ is good itself or $(P_{-k},p_k')$ is good with
  probability at least $1-2n^3/\ell$, where $p_k'$ is a random sibling
  of $p_k$ (Lemma~\ref{lemma:many-good-sets}). Roughly speaking,
  either $P$ is good or {\em almost all other sets are good} (with
  exponentially small probability of the negative event).
\end{description}

By Assumption~\ref{ass1}, there is no bad task, hence $b_1=0$, and then in
Lemma~\ref{lemma:prob-of-bad} we show how to combine these two facts in order to
get that $b_k$ is bounded above by $3n^3/\ell(3n/\delta)^{k-2}$
(Lemma~\ref{lemma:prob-of-bad}). If we allow the number $\ell+1$ of tasks per
cluster to be sufficiently large, this shows that $b_{n-1}<1$.

\paragraph{Showing Fact b.}

The difficult part is to establish the second of the above two facts
(Fact b). Let's assume that $P$ is potentially-good but not good. We
show that $(P_{-k},p_k')$ is good for many $p_k'$'s, as follows
\begin{itemize}
\item first, we observe that there is a witness of potential goodness of $P$ for which all tasks in $P=(p_1,\ldots,p_k)$ are given to the $s$-players. This essentially follows from the definition of potentially-good and weak monotonicity (Lemma~\ref{lemma:all-tasks-to-s})
\item let $p_k'$ be a sibling of $p_k$ and consider the
  $(p_k,p_k')$-slice mechanism. This is exactly the point where we
  exploit the $2\times 2$ characterization that we provide in
  Section~\ref{sec:char-all-truthf}. The proof proceeds by treating
  carefully all possible cases
  \begin{description}
  \item{\bf affine minimizers:} we show that the mechanism is not an affine
    minimizer almost surely (Lemma~\ref{lemma:prob-of-linear}) and the stronger
    claim that they do not exist (Lemma~\ref{lemma:prob-of-linear-strong})
  \item{\bf relaxed affine minimizers:} we show that the probability that
    the mechanism is a relaxed affine minimizer is at most $2n^2/\ell$
    (Lemma~\ref{lemma:prob-of-relaxed-aff-min})
  \item{\bf 1-dimensional and constant mechanisms:} we show that
    1-dimensional and constant mechanisms do not occur, otherwise
    the approximation ratio is high (Lemma~\ref{lemma:prob-of-1dim})
  \item{\bf task independent or relaxed task independent mechanisms:} we show that
    if the mechanism is task independent or relaxed task independent for each of
    $k$ appropriately selected random instances from the witness, then
    $(P_{-k},p_k')$ is good (Lemma~\ref{lemma:magic-task-independent})
  \end{description}
\item we conclude that for a random sibling $p_k'$, the mechanism must
  be either task independent or relaxed task independent for all
  these $k$ instances with probability at least $1-k\, 2n^2/\ell$;
  therefore $(P_{-k},p_k')$ is good with probability at least
  $1-2n^3/\ell$.  (Lemma~\ref{lemma:many-good-sets}).
\end{itemize}
The first item, i.e., to show that affine minimizers are sparse, exploits an interesting use of goodness and linearity. This is where the strange second point of Definition~\ref{def:potentially-good}  is needed. The proof of relaxed affine minimizers uses the same machinery, but it has an extra layer of difficulty, which arises from having to guarantee that the action happens at the linear part and not at the bundling tail of these mechanisms. In fact, we might have a positive probability (at most $2n^3/\ell$) to pick a wrong sibling $p_k'$ solely due to boundling tails of relaxed affine minimizers. The proof of the last item about task independent and relaxed task independent mechanisms has very similar flavor. It is essentially this part that takes away the complications that arise from having to deal with an additive domain.

The core of the argument for (relaxed) affine minimizers uses the characterization of the $(p_k,p_k')$-slice mechanism in a way that seems peculiar at first glance. Instead of focusing on the $(p_k,p_k')$-slice mechanism, it focuses on the slice $(p_i,p_k)$, for some $p_i\in P_{-k}$ (as in Figure~\ref{fig:magic-linear}). The reason is simple: the characterization is only used to extract the property that the allocation of $p_k$ has linear boundaries, and then taking advantage of the fact ---due to potentially-goodness property--- that the $(p_i,p_k)$-slice mechanism exhibits a tight connection between the allocation boundaries of $p_i$ and $p_k$, we can argue about the allocation of task $p_i$. A similar argument is used when the $(p_k,p_k')$-slice mechanism is task independent.

\subsection{General lemmas}
\label{sec:general-lemmas}

We start by establishing some useful facts. The following simple lemma establishes that mechanisms for slices of two tasks that are affine minimizers or relaxed affine minimizers have coefficients bounded by the approximation ratio. This is useful because when we change an $s$-value by a given amount, we can bound the change in the boundary for the $t$-value.

\begin{lemma} \label{lemma:2x2-aff-min} Let $P=\{p_1,\ldots,p_k\}$ be a
  potentially-good set of tasks for $T$ with witness $V$ and let $T'$ be an
  instance in the set of $V$-perturbations of $\hat T(P)$. Consider a sibling
  $p_k'$ of $p_k$ and assume that the $(p_k,p_k')$-slice mechanism for $T'$ is
  an affine minimizer or a relaxed affine minimizer with boundary function
  $\psi_{p_k}(s_{p_k})$ for task $p_k$, which is linear when we fix the values of
  all other tasks i.e., $\psi_{p_k}(s_{p_k})=\max(0, \lambda\, s_{p_k} - \gamma)$ for some $\lambda$ and $\gamma$ that may depend on all other values but not on $t_{p_k}$ and $s_{p_k}$. Suppose further that $\psi_{p_k}(1)>0$. Then the approximation ratio of the mechanism (for all tasks) is at least $1+\max(\lambda, 1/\lambda)-\delta'$.
\end{lemma}
\begin{proof}
  First notice that $\gamma$ has small value, because by the premises of the lemma, $\psi_{p_k}(1)=\lambda\cdot 1-\gamma > 0$, which gives $\gamma\leq \lambda$.
  
  If $\lambda\geq 1$, then set $t_{p_k}=\lambda s_{p_k}-\gamma -
  \epsilon$, for some arbitrarily small $\epsilon>0$ and observe that
  the task is given to the 0-player. We argue that, since the
  contribution of the other tasks is limited and $\gamma$ is small,
  there exists a sufficiently large value\footnote{We cannot take
    arbitrarily large values for $s_{p_k}$, since we assume that the
    domain is bounded above by a very large constant $B$ (see
    Section~\ref{sec:char-all-truthf}).} for $s_{p_k}$, such that the
  approximation ratio is at least $\lambda-\epsilon$. The values of
  the other tasks contribute at most $(n-1)(\alpha+\delta')\leq n$ to
  the optimal cost. Task $p_k$ is assigned to the $0$-player, but the
  optimal allocation would be to give it to the $s$-player, with an
  approximation ratio almost $\lambda$.

  We can improve this to $1+\lambda$ by using the dummy tasks as
  follows. We change the $t$-value of dummy task $d_0$ to make it
  equal to $s_{p_k}$ and lower slightly the $t$-value of task
  $p_k$. Assuming that the allocation of the dummy task does not
  change, by weak monotonicity (Lemma~\ref{lemma:tool}), the
  allocation of the 0-player remains the same. The cost of the
  mechanism is at least $s_{p_k}+t_{p_k}\approx (1+\lambda)s_{p_k}$,
  while the optimal allocation is at most
  $s_{p_k}+n+\lambda\approx s_{p_k}$, 
  which gives an approximation ratio very close to $1+\lambda$; for
  sufficiently high $s_{p_k}$, at least $1+\lambda-\delta'$.

  The case of $\lambda < 1$ is similar, but we set $t_{p_k}=\lambda
  s_{p_k}-\gamma + \epsilon$, and we change the $s$-value of the dummy
  task $d_k$ to $t_{p_k}$. In this case, we get a lower bound for
  approximation ratio $1+1/\lambda-\delta'$.
\end{proof}
We assume that the mechanism has approximation ratio less than $\rho$, which means that if a $(p,p')$-slice mechanism is an affine minimizer or a relaxed affine minimizer, the coefficient $\lambda$ of its linear boundary is in the interval $[1/(\rho-1+\delta'), \rho-1+\delta']$. This guarantees that for some task $p_k$, if the $s_{p_k}$ changes by $\Delta s_{p_k}$, then the boundary for $t_{p_k}$ changes by at least $\Delta s_{p_k}/(\rho-1+\delta')$ and at most $(\rho-1+\delta')\Delta s_{p_k}$.  We will take advantage of this fact in Lemma~\ref{lemma:magic-affine minimizer}, where we will combine this with the second condition in the definition of potentially-good set of tasks. This second condition states that $P_{-k}$ is good for instances in which task $p_k$ has values $[t_{p_k}=\delta, \, s_{p_k}=q \delta/(2n)]$, for $q=0,1,\ldots,2n/\delta$. For successive values of $q$, $s_{p_k}$ changes by $\delta/(2n)$. The above lemma allows us to conclude that successive values of the boundary of $t_{p_k}$ change by at most $\rho \delta/(2n) \leq \delta/2$ and that one of these values is less than $\delta$ (see the proof of Lemma~\ref{lemma:magic-affine minimizer} for details). This goal lies behind the complicated definition of potentially-good set of tasks.

The following lemma says that a proper allocation of a single instance can be extended by weak monotonicity to a set of perturbations. It is used in two parts of the proof. First immediately below, to show that all tasks of a potentially-good but not good set is allocated to the $s$-players (Lemma~\ref{lemma:all-tasks-to-s}). And later on to deal with task independent or relaxed task independent slice mechanisms (Lemma~\ref{lemma:magic-task-independent}).
 
\begin{lemma} \label{lemma:new-witness} Let $T$ be a standard instance for a regular set of tasks $P=(p_1,\ldots,p_k)$. Let $T'=(t',s)$ be an instance with $t_{p_i}'\in(\alpha,\alpha+\beta)$, $i\in [k]$, which agrees with $T$ on the remaining tasks. If all tasks in $P$ are allocated to the 0-player in the allocation of $T'$, then there exists $V$ such that $P$ is good for $T$ with witness $V$.
\end{lemma}
\begin{proof}
  Define $V$ such that $V_{p_i}=(\alpha, t'_{p_i})$ for each task $p_i\in P$. By weak monotonicity (Lemma~\ref{lemma:tool}), the 0-player is allocated all tasks in $P$ for every instance in the set of $V$-perturbations of $\hat T(P)$. Therefore $P$ is good for $T$ with witness $V$.
\end{proof}

The next lemma shows that all tasks of a potentially-good but not good set $P$ of tasks must be allocated to the $s$-players. Furthermore, $P$ is a minimal or critical set in the sense that if we lower the $t$-value of one of its tasks, the allocation for all of them changes. The first part of the definition of potentially-good sets of tasks is designed to achieve exactly this goal.

\begin{lemma} \label{lemma:all-tasks-to-s} Suppose that $P=(p_1,\ldots,p_k)$ is a potentially-good set of tasks for $T$ with witness $V.$ If $P$ is not good for $T$ and there exists no bad task, then for every instance of the set of $V$-perturbations of $\hat T(P)$, all tasks in $P$ are allocated to the $s$-players. Furthermore, for each $i\in [k]$, if we change the $t$-value of only task $p_i$ to $\beta$, all tasks in $P$ are allocated to the $0$-player.
\end{lemma}
\begin{proof}
  Fix an instance $T'$ in the set of $V$-perturbations of $\hat T(P)$,
  and let $D$ be the subset of $P$ that is allocated to the
  0-player. We first show that we cannot have $\emptyset\subsetneq
  D\subsetneq P$.

  Fix a task $p_r\in D$ and change the $t$-values of $T'$ as follows:
  \begin{itemize}
  \item task $p_r$ is set to $\beta$
  \item all other tasks in $D$ are slightly decreased\footnote{The changes in the instance are such that it remains in the set of $V$-perturbations of $\hat T(P_{-r})$.}
  \item tasks in $P\setminus D$ are slightly increased
  \end{itemize}
  By weak monotonicity (Lemma~\ref{lemma:tool}), the allocation of $P$ remains the same. In particular, all tasks in $P\setminus D\neq\emptyset$ are allocated to the $s$-players. But this contradicts the assumption that $P_{-r}$ is good.

  So it must be that the mechanism for $T'$ allocates either all or none of the tasks in $P$ to the 0-player. In the first case, we apply Lemma~\ref{lemma:new-witness}, to get that $P$ is a good set of tasks for $T$ with some witness $V'$, which contradicts the premises of the lemma. So it must be that all tasks are allocated to the $s$-players.

  Suppose now that we change the $t$-value of some $p_i$ to $\beta$. Since $P_{-i}$ is good, all tasks in $P_{-i}$ are allocated to the 0-player. We now show that task $p_i$ is also allocated to the 0-player. For if not, we can lower all $t$-values in $P_{-i}$ to $\beta$, and increase the $t$-value of $p_i$ to $\alpha+\epsilon$, for  $\epsilon\in(0,\beta)$. By weak monotonicity (Lemma~\ref{lemma:tool}), the allocation of the tasks in $P$ remains the same. But then $\{p_i\}$ is a bad task, a contradiction.
\end{proof}

\subsection{Almost all sets are good}
\label{sec:slice}
We now consider a potentially-good but not good set of tasks
$P=(p_1,\ldots,p_k)$ from clusters $(Q_1,\ldots,Q_k)$.  Let $p_k'$ be
a sibling of $p_k$, i.e., another task of cluster $Q_k$.  The main
result of this section establishes that {\em almost all other sets
  $(P_{-k},p_k')$ are good}; we show via a probabilistic argument that
$(P_{-k},p_k')$ is good with probability at least $1-2n^3/\ell$, where
$p_k'$ is a random sibling of $p_k$ (Lemma~\ref{lemma:many-good-sets}).

In order to show this result, we consider the $(p_k,p_k')$-slice
mechanism, utilizing the $2\times 2$ characterization that we provide
in Section~\ref{sec:char-all-truthf}, considering all possible cases.
Specifically, in Section~\ref{sec:affine-or-relaxed} we exclude
(almost surely) affine minimizers (Lemma~\ref{lemma:prob-of-linear}),
in Section~\ref{sec:relax-affine-minim} we bound from above the
probability that the mechanism is a relaxed affine minimizer by
$2n^2/\ell$ (Lemma~\ref{lemma:prob-of-relaxed-aff-min}) and in Section
\ref{sec:1-dim-constant} we exclude 1-dimensional and constant
mechanisms (Lemma~\ref{lemma:prob-of-1dim}). Finally in Section
\ref{sec:indep-or-relax} we show that if the mechanism is
task independent or relaxed task independent for each of $k$
appropriately selected random instances from the witness, then
$(P_{-k},p_k')$ is good (Lemma~\ref{lemma:magic-task-independent}).

We conclude that for a random sibling $p_k'$, the mechanism must
  be either task independent or relaxed task independent for all
  these $k$ instances with probability at least $1-k\, 2n^2/\ell$;
  therefore $(P_{-k},p_k')$ is good with probability at least
  $1-2n^3/\ell$.  (Lemma~\ref{lemma:many-good-sets}).

\subsubsection{Affine minimizers}
\label{sec:affine-or-relaxed}

The following lemma (Lemma~\ref{lemma:magic-affine minimizer}) deals
with the case that the $(p_k,p_k')$-slice mechanism is an affine
minimizer or in the linear part of a relaxed affine minimizer. It is
essential in showing that {\em very few} slice-mechanisms are of this
type. We also provide a stronger version of this statement in
Lemma~\ref{lemma:magic-affine minimizer - strong}, which states that
linear slice mechanisms do not exist at all. We include both
Lemma~\ref{lemma:magic-affine minimizer} and its stronger version,
because the proof of Lemma~\ref{lemma:magic-affine minimizer} is
easier to follow, especially when its main argument is illustrated by
a 2-dimensional figure, while the stronger version provides an
interesting algebraic generalization of the proof. 

\begin{lemma} \label{lemma:magic-affine minimizer} Assume that there is no bad task and that the approximation ratio is at most $\rho$. Suppose that $P=(p_1,\ldots,p_k)$ is a potentially-good but not good set of tasks for $T$ with witness $V$. There are no instances $T_1$, $T_2$ such that
  \begin{itemize}
  \item both $T_1$ and $T_2$ belong to the set of $V$-perturbations of $\hat T(P)$
  \item $T_1$ and $T_2$ differ only in the value of $t_{p_i}$, for some fixed $i\in [k-1]$
  \item for both $T_1$ and $T_2$, the boundary function $\psi_{p_k}$ is truncated linear in $s_{p_k}$. That is,
    \begin{align*}
      \psi_{p_k}(t_{-p_k},s)=\max(0, \lambda(
      t_{-p_k},s_{-p_k}) s_{p_k} - \gamma(t_{-p_k},s_{-p_k})),
    \end{align*}
    for both $T_1$ and $T_2$.
  \end{itemize}
\end{lemma}
\begin{proof}
  Towards a contradiction, suppose that such instances $T_1$ and $T_2$ exist. Let $\alpha'<\alpha''$ be the two different values of $t_{p_i}$ for $T_1,T_2$, and $\alpha_k$ be the value of $t_{p_k}$. Focus on the $(p_i, p_k)$-slice mechanism and observe that there is a bundling boundary between tasks $p_i,p_k$ when $t_{p_i}\in (\alpha',\alpha'')$ (see Figure~\ref{fig:magic-linear}). This is because the allocation for both $(t_{p_i},t_{p_k})=(\alpha',\alpha_k)$ and $(t_{p_i},t_{p_k})=(\alpha'',\alpha_k)$ gives both $p_i,p_k$ to the $s$-player, while the allocation for both $(t_{p_i},t_{p_k})=(\alpha',\beta)$, $(t_{p_i},t_{p_k})=(\alpha'',\beta)$ gives both $p_i,p_k$ to the $0$-player (Lemma~\ref{lemma:all-tasks-to-s}).
  
  Let us focus on the dependency of boundary $\psi_{p_k}$ of task $p_k$ on $t_{p_i}$, and $s_{p_k}$, since every other value is fixed. For simplicity we write $\psi_{p_k}(t_{p_i},s_{p_k})$ for this boundary, ignoring the other fixed values. By the premises of the lemma, this boundary is a truncated linear function for both values $\alpha',\alpha''$ of $t_{p_i}$:
  \begin{align*}
    \psi_{p_k}(\alpha',s_{p_k}) &= \max(0, \lambda(\alpha') s_{p_k} - \gamma(\alpha'))\\
    \psi_{p_k}(\alpha'',s_{p_k})&= \max(0, \lambda(\alpha'') s_{p_k} - \gamma(\alpha'')),
  \end{align*}
  where the $\lambda$'s and $\gamma$'s are independent of the values of task $p_k$. We now argue that $\lambda(\alpha')=\lambda(\alpha'')$. This implies that when we decrease $s_{p_k}$, the boundary of $\psi_{p_k}(t_{p_i},s_{p_k})$, as a function of $t_{p_i}$, simply slides left rectilinearly; in Figure~\ref{fig:magic-linear}, the blue sloped line in the interval $[\alpha',\alpha'']$ moves left and remains sloped (until it reaches the $t_{p_i}$ axis).

  This is a simple, essentially geometric, argument; Lemma 4 in \cite{CKK20} has a similar argument. First observe that $\psi_{p_k}(\alpha'',1)-\psi_{p_k}(\alpha',1)=-(\alpha''-\alpha')$, since for $s_{p_k}=1$ we have argued that there is a bundling boundary between tasks $p_i,p_k$ when $t_{p_i}\in (\alpha',\alpha'')$. Suppose first that $\lambda(\alpha')>\lambda(\alpha'')$, and observe that for every fixed $s_{p_k}$, $\psi_{p_k}(t_{p_i},s_{p_k})$ is piecewise linear with specific slopes (Figure~\ref{fig:magic-linear}); in particular the derivative $\partial \psi_{p_k}(t_{p_i},s_{p_k})/\partial t_{p_i}$ of the linear parts is in $\{-1,0,1\}$. When $s_{p_k}=1+\theta$, for any positive $\theta$, we get
  \begin{align*} \psi_{p_k}(\alpha'',1+\theta)-\psi_{p_k}(\alpha',1+\theta) &=\psi_{p_k}(\alpha'',1)-\psi_{p_k}(\alpha',1)+(\lambda(\alpha'') - \lambda(\alpha'))\theta<-(\alpha''-\alpha'),
  \end{align*}
  which is impossible because the derivative $\partial \psi_{p_k}(t_{p_i},s_{p_k})/\partial t_{p_i}$ cannot be less than $-1$. The case of $\lambda(\alpha')<\lambda(\alpha'')$ is similar, but we use $s_{p_k}=1-\theta$.

  For the rest of the proof, we simply write $\lambda$ instead of
  $\lambda(\alpha'')$, and $\gamma$ instead of $\gamma(\alpha'')$ and focus on
  $t_{p_i}=\alpha''$. We will show that when we lower $s_{p_k}$ to some
  $\epsilon>0$, task $p_k$ is allocated to the $s$-player for any value
  $t_{p_k}\geq \beta$.  We use the fact that $\lambda\geq 1/(\rho-1+\delta')
  \geq \alpha$ (Lemma~\ref{lemma:2x2-aff-min}), which intuitively shows that the
  boundary moves at relative speed at least $\alpha$ as we move $s_{p_k}$ to
  0. Since task $p_k$ is given to player $k$ for $[t_{p_k}=\alpha_k,\,
  s_{p_k}=1]$, we must have $\psi_{p_k}(\alpha'',1)\leq \alpha_{k}$. This gives
  $\lambda\cdot 1 - \gamma\leq \alpha_k$, and we get
  $\psi_{p_k}(\alpha'',\epsilon)\leq \max(0,\lambda\cdot\epsilon -\gamma)\leq
  \max(0, \lambda\cdot\epsilon+\alpha_k-\lambda)\leq \max(0,
  \lambda\cdot\epsilon+\alpha_k- \alpha) < \max(0, \beta)= \beta$, when
  $\epsilon$ is sufficiently small.
  \begin{figure}
    \centering
    \begin{tikzpicture}

      \draw[->] (0,0) -- (6,0) node[anchor=north] {$t_{p_k}$};
      \draw[->] (0,0) -- (0,6) node[anchor=east] {$t_{p_i}$};
      \draw[very thick, blue] (0,3.5) -- (2,3.5) -- (4,1.5) -- (6,1.5) ;

      \draw[ blue, dashed] (0,2.5) node[anchor=east] {$\alpha'$} -- (6,2.5) ;

      \draw[ blue, dashed] (0,3.3) node[anchor=east] {$\alpha''$} -- (6,3.3) ;

      \draw[ blue, dashed] (3.4,0) node[anchor=north] {$\alpha_k$} -- (3.4,6) ; \draw[ blue, dashed] (0.7,0) node[anchor=north] {$\delta$} -- (0.7,6) ;

      \draw[very thick, blue] (4,0) -- (4, 1.58) -- (2, 3.58) -- (2, 6) node[anchor=south] {$\psi_{p_k}(t_{p_i}, 1)$} ;

      \draw[ultra thick,black] (2.2,3.3) circle (1.mm); \draw[ultra thick,black] (3,2.5) circle (1.mm);

      \draw[ultra thick, red] (1.3, 2.52) -- (0.5, 3.32) ; \draw[very thick, red, dashed] (1.6, 2.52) -- (0.8, 3.32) ; \draw[very thick, red] (2.0, 2.52) node[anchor=north] {$\psi_{p_k}(t_{p_i}, q^*\delta/(2n)$} ; \draw[very thick, red,dashed] (1.8, 2.52) -- (1, 3.32) ; \draw[very thick, red,dashed] (2, 2.52) -- (1.2, 3.32) ; \draw[very thick, red,dashed] (2.3, 2.52) -- (1.5, 3.32) ; \draw[very thick, red,dashed] (2.5, 2.52) -- (1.7, 3.32) ; \draw[very thick, red,dashed] (2.7, 2.52) -- (1.9, 3.32) ;

    \end{tikzpicture}
    \caption{\small Illustration for the proof of Lemma~\ref{lemma:magic-affine minimizer}. The blue figure shows the allocation of the $0$-player for tasks $p_{i},p_{k}$, when the values of all tasks (except for $p_i,p_k$) agree with the values in instances $T_1,T_2$. The black circles show the bundling boundary points $\psi_{p_k}(\alpha',1)$, $\psi_{p_k}(\alpha'',1)$. The red part of the figure demonstrates that the boundary points $\psi_{p_k}(\alpha',s_{p_k})$, $\psi_{p_k}(\alpha'',s_{p_k})$ (and hence all the boundary points in the band $t_{p_i}\in(\alpha',\alpha'')$) must move rectilinearly when we decrease $s_{p_k}$. The thick red boundary line corresponds to $s_{p_k}=q^*\delta/(2n)$, for which the boundary point at $t_{p_i}=\alpha''$ has value smaller than $\delta$. Since the point $(\delta, \alpha'')$ is above the thick red boundary, both task $p_i$ and $p_k$ are given to the $s$-players, for any value $t_{p_k}\geq \delta$, when $s_{p_k}=q^*\delta/(2n)$. The fact that $p_i$ is not given to the 0-player for these values contradicts the goodness of $P_{-k}$ for these values.}
    \label{fig:magic-linear}
  \end{figure}
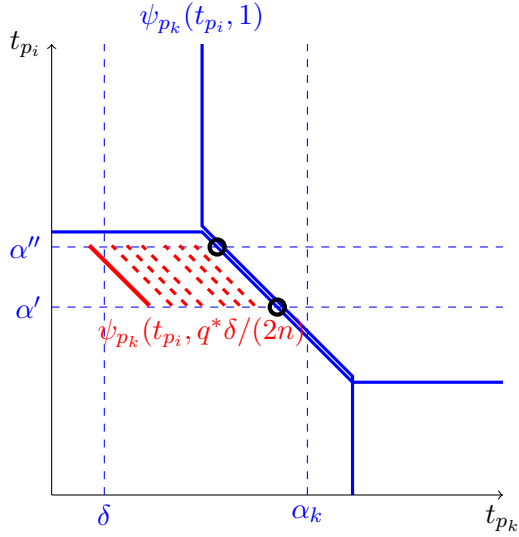
  We now consider the sequence of values $\psi_{p_k}(\alpha'',q \delta/(2n))$, for $q=0,1,\ldots,2n/\delta$. By Lemma~\ref{lemma:2x2-aff-min}, $\lambda\leq \rho-1+\delta'$, so successive values in this sequence differ by at most $\lambda \delta/(2n)\leq (\rho-1+\delta') \delta/(2n)\leq \delta/2$. The previous paragraph established that one of these values is at most $\beta$. Since the maximum value is $\psi_{p_k}(\alpha'',1)\geq \beta$, one of these values, say value $\psi_{p_k}(\alpha'',q^* \delta/(2n))$ is in the interval $(\beta,\delta)$. As a result, task $p_k$ for values $[t_{p_k}=\delta'', \, s_{p_k}=q^* \delta/(2n)]$ is given to the $k$-player, for every $\delta''\geq \delta$. Since this remains above a bundling boundary between $p_i$ and $p_k$, task $p_i$ is also given to the $k$-player (see Figure~\ref{fig:magic-linear} for an illustration). But this contradicts the goodness of $P_{-k}$ for $[t_{p_k}=\delta, \, s_{p_k}=q^* \delta/(2n)]$.
\end{proof}

The above lemma is crucial and its importance is captured by the next
lemma that essentially states that linear $2\times 2$ mechanisms do
not exist almost surely.

\begin{lemma} \label{lemma:prob-of-linear} Suppose that $P=(p_1,\ldots,p_k)$ is a potentially-good but not good set of tasks for $T$ with witness $V$. Take a random instance $T'$ from the set of $V$-perturbations of $\hat T(P)$. Fix all other values to $T'$ except for task $p_k$, and consider the event $\cal E$ that the boundary function $\psi_{p_k}(s_{p_k})$ is linear in $s_{p_k}$. The probability of event $\cal E$ is 0.
\end{lemma}
\begin{proof}
  By Lemma~\ref{lemma:magic-affine minimizer}, there are no two instances $T_1$ and $T_2$ that differ in only one task of $P_{-k}$ for which event $\cal E$ happens. Using the notation $\hat V=\times_{j\in P} V_j$ and $\hat V_{p_1}=\times_{j\in P_{-1}} V_j$, the probability of event $\cal E$ is at most $\vol_k(\hat V_{p_1})/\vol_k(\hat V)=0$, where $\vol_k$ denotes the $k$-dimensional volume, because $\vol_k(\hat V)>0$ and $\vol_k(\hat V_{p_1})=0$ since $\hat V_{p_1}$ is a hyperrectangle of $k-1$ dimensions.
\end{proof}

\subsubsection*{A stronger version of Lemma~\ref{lemma:magic-affine minimizer}}
\label{sec:altern-proof-lemma}

Lemma~\ref{lemma:prob-of-linear} shows that affine minimizers, or more generally
mechanisms with linear boundaries, have probability 0. Can we exclude them
completely? Indeed we can, by strengthening Lemma~\ref{lemma:magic-affine
  minimizer} as follows. Instead of excluding linear mechanisms in instances
that differ in only one task, the stronger lemma excludes linear mechanisms in
two distinct instances $T_1$, $T_2$ that can differ in all tasks in $P$ and
$T_2\geq T_1$. This shows that there exists a perturbation with no affine
minimizers.

\begin{lemma} \label{lemma:magic-affine minimizer - strong} Assume that there is no bad task and that the approximation ratio is at most $\rho$. Suppose that $P=(p_1,\ldots,p_k)$ is a potentially-good but not good set of tasks for $T$ with witness $V$. There are no distinct instances $T_1$, $T_2$ such that
  \begin{itemize}
  \item both $T_1$ and $T_2$ belong to the set of $V$-perturbations of $\hat T(P)$
  \item the value of $t_{p_i}$ in $T_1$ is at most equal to the value of $t_{p_i}$ in $T_2$, for all $i\in [k-1]$, i.e., $T_1 \leq T_2$
  \item for both $T_1$ and $T_2$, the boundary function $\psi_{p_k}$ is linear in $s_{p_k}$. That is,
    \begin{align*}
      \psi_{p_k}(t_{-p_k},s)=\max(0, \lambda(t_{-p_k},s_{-p_k}) s_{p_k} - \gamma(t_{-p_k}, s_{-p_k})),
    \end{align*}
    for both $T_1$ and $T_2$.
  \end{itemize}
\end{lemma}
\begin{proof}
  Towards a contradiction, suppose that such instances $T_1=(t^1,s)$ and $T_2=(t^2,s)$ exist. For simplicity we write $\psi_{p_k}(t_{-p_k}^j,s_{p_k})$ for the boundary, ignoring the other fixed $s$-values. By Lemma~\ref{lemma:all-tasks-to-s}, we know that in the allocation of $T_j$, $j=1,2$, all tasks are given to the $s$-player and when we change one $t$-value of one of the tasks in $P$ to $\beta$, all tasks are given to the 0-player.

  We break the proof into simpler claims. First there is a very general claim that the boundary functions are 1-Lipschitz.

  \newcounter{claimNo}

  \stepcounter{claimNo}
  \begin{claim*}[\roman{claimNo}]
    The boundary function $\psi_r(t_{-r},s)$ is 1-Lipschitz in $t_{-r}$, i.e.,
    $|\psi_r(t_{-r},s)-\psi_r(t_{-r}',s)|\leq |t_{-r}-t_{-r}'|_1$.
  \end{claim*}
  \begin{proof}
    
    Consider the instances $(t_{-r},\psi_r(t_{-r},s)-\epsilon)$ and $(t_{-r}',\psi_r(t_{-r}',s)+\epsilon)$. By the definition of the boundary function $\psi_r$, task $r$ is allocated to the player in the first instance, but not in the second one.  Let $a$ and $a'$ denote the allocation for these two instances. Then by weak monotonicity, we have
\begin{align*}
  \sum_{i\neq r} (a_i-a_i') (t_i-t_i')+(1-0)((\psi_r(t_{-r},s)-\epsilon)-(\psi_r(t_{-r}',s)+\epsilon)) \leq 0 \\
  \psi_r(t_{-r},s)-\psi_r(t_{-r}',s) \leq \sum_{i\neq r} (a_i'-a_i) (t_i-t_i') + 2\epsilon.%
\end{align*}
Since $(a_i'-a_i)\in\{0,1,-1\}$, by letting $\epsilon$ tend to 0, we get $\psi_r(t_{-r},s)-\psi_r(t_{-r}',s)\leq |t_{-r}-t_{-r}'|_1$. The other direction follows by the symmetry of $t$ and $t'$.
\end{proof}

\stepcounter{claimNo}
\begin{claim*}[\roman{claimNo}]  $\psi_{p_k}(t_{-p_k}^2,1)-\psi_{p_k}(t_{-p_k}^1,1)=-|t_{-p_k}^2-t_{-p_k}^1|_1=-\sum_{i\in [k-1]} (t_{p_i}^2-t_{p_i}^1).$
  \end{claim*}
  \begin{proof}
    We know that in the allocation of instance $T_1$ all tasks in $P$ are given to the $s$-players. We create another instance $\hat t^1$ that has the same allocation for the tasks in $P$: we increase each $t_{p_i}^1$ by $\epsilon/n$ and set $t_{p_k}^1=\psi_{p_k}(t_{-p_k}^1,1)+\epsilon$, for some small $\epsilon>0$. The allocation of $P$ remains the same, because (1) the boundary $\psi_{p_k}(t_{-p_k}^1,1)$ changes by less than $(n-1)\epsilon/n<\epsilon$, so the new $t$-value of $p_k$ is still above the boundary, and (2) by weak monotonicity, the allocation of the other tasks remains the same. Similarly, we create an instance $\hat t^2$ that is very close to $t^2$ except that $t_{p_k}^2=\psi_{p_k}(t_{-p_k}^2,1)-\epsilon$ and all tasks are allocated to the 0-player: start with $(t_{-p_k}^2,\beta)$ for which all tasks are allocated to the 0-player, and then decrease each $t_{p_i}^2$ by $\epsilon/n$ and set $t_{p_k}^2=\psi_{p_k}(t_{-p_k}^2,1)-\epsilon$. Again the allocation will not change. The weak monotonicity property gives
  \begin{align*}
    \sum_{i\in [k]} (1-0)(\hat t_{p_i}^2-\hat t_{p_i}^1)&\leq 0 \\
    \sum_{i\in [k-1]} (t_{p_i}^2-t_{p_i}^1) + (\psi_{p_k}( t_{-p_k}^2,1)-\psi_{p_k}( t_{-p_k}^1,1))  &\leq 4\epsilon \\
    \psi_{p_k}( t_{-p_k}^2,1)-\psi_{p_k}( t_{-p_k}^1,1) &\leq -\sum_{i\in [k-1]} (t_{p_i}^2-t_{p_i}^1) +4\epsilon.
  \end{align*}
By letting $\epsilon$ tend to 0 and by the Lipschitz property, we get the equality of the claim. To use the Lipschitz property, we need to recall that $T_2\geq T_1$, so all terms inside the sum on the right hand side are nonnegative.
\end{proof}

\stepcounter{claimNo}
\begin{claim*}[\roman{claimNo}]
\newcounter{claim3}\setcounter{claim3}{\value{claimNo}}

The boundary functions for $t^j$, $j=1,2$,
  \begin{align*}
      \psi_{p_k}( t_{-p_k}^j,s_{p_k}) & =\lambda(t_{-p_k}^j) s_{p_k} - \gamma(t_{-p_k}^j)
  \end{align*}
  have the same $\lambda$ coefficient, i.e., $\lambda(t_{-p_k}^1)=\lambda(t_{-p_k}^2)=\lambda$. As a result
  \begin{align*}
    \psi_{p_k}( t_{-p_k}^2,s_{p_k})-\psi_{p_k}( t_{-p_k}^1,s_{p_k})=-|t_{-p_k}^2-t_{-p_k}^1|_1,
  \end{align*}
  for every $s_{p_k}$, for which the boundaries are positive.
\end{claim*}
\begin{proof}
  Suppose not and assume that $\lambda(t_{-p_k}^1)>\lambda(t_{-p_k}^2)$. For $s_{p_k}=1+\epsilon$, for some $\epsilon>0$, we get
  \begin{align*}
    \psi_{p_k}( t_{-p_k}^2,1+\epsilon)-\psi_{p_k}( t_{-p_k}^1,1+\epsilon)&=\psi_{p_k}( t_{-p_k}^2,1)-\psi_{p_k}( t_{-p_k}^1,1)+(\lambda(t_{-p_k}^2)-\lambda(t_{-p_k}^1))\epsilon \\
    &<-|t_{-p_k}^2-t_{-p_k}^1|_1,
  \end{align*}
  which violates the Lipschitz property. The case of $\lambda(t_{-p_k}^1)<\lambda(t_{-p_k}^2)$ is similar, only that we now consider $s_{p_k}=1-\epsilon$.
\end{proof}

\stepcounter{claimNo}
\begin{claim*}[\roman{claimNo}]
  There is $\epsilon>0$ such that $\psi_{p_k}(t_{-p_k}^j,\epsilon)\leq \beta$, for $j=1,2$.
\end{claim*}
\begin{proof}
  We use the fact that $\lambda\geq 1/(\rho-1+\delta') \geq \alpha$
  (Lemma~\ref{lemma:2x2-aff-min}), which intuitively shows that the boundary
  moves at relative speed at least $\alpha$ as $s_{p_k}$ changes from 1 to
  0. Since task $p_k$ is given to the $s$-player for $T_j$, we must have
  $\psi_{p_k}(t_{-p_k}^j,1)\leq t_{p_k}^j < \alpha+\beta$. This gives
  $\lambda\cdot 1 - \gamma(t_{-p_k}^j) < \alpha+\beta$, and for a sufficiently
  small $\epsilon>0$, we get
  $\psi_{p_k}(t_{-p_k}^j,\epsilon) \leq \max(0, \lambda\cdot \epsilon-\gamma(t_{-p_k}^j))\leq
  \max(0, \alpha+\beta-\lambda)\leq \max(0, \beta)\leq \beta$.
\end{proof}

\stepcounter{claimNo}
\begin{claim*}[\roman{claimNo}]
  There exists $q^*\in \{0,1,\ldots,2n/\delta\}$ such that $\psi_{p_k}(t_{-p_k}^j,q^* \delta/(2n))$ is in the interval $[\beta,\delta)$, for $j=1,2$.
\end{claim*}
\begin{proof}
  Consider the sequence of values $\psi_{p_k}(t_{-p_k}^j,q \delta/(2n))$, for $q=0,1,\ldots,2n/\delta$. By Lemma~\ref{lemma:2x2-aff-min}, $\lambda\leq \rho-1+\delta'$, so successive values in this sequence differ by at most $\lambda \delta/(2n)\leq (\rho-1+\delta') \delta/(2n)\leq \delta/2$. The previous paragraph established that one of these values is at most $\beta$, while the value for $q=2n/\delta$ is at least $\beta$. Therefore one of these values, say value $\psi_{p_k}(t_{-p_k}^j,q^* \delta/(2n))$ is in the interval $[\beta,\delta)$. 
\end{proof}

Let us denote by $\delta^j$ the value $\psi_{p_k}(t_{-p_k}^j,q^* \delta/(2n))$
of the last claim. By Claim~(\roman{claim3}) these boundary values satisfy $\delta^2-\delta^1=-|t_{-p_k}^2-t_{-p_k}^1|_1$.  

We now have all ingredients to finish the proof of the lemma. Fix some arbitrarily small $\epsilon>0$. We create two new instances: instance $\hat T_1=(\hat t^1,\hat s)$ that agrees with $T_1$ everywhere except that $[\hat t_{p_k}^1=\delta^1-\epsilon,\, \hat s_{p_k}=q^* \delta/(2n)]$; and instance $\hat T_2=(\hat t^2,\hat s)$ that agrees with $T_2$ everywhere except that \emph{each $t_{p_i}^2$ is decreased by $\epsilon/n$} and $[\hat t_{p_k}^2=\delta^2+\epsilon,\, \hat s_{p_k}=q^* \delta/(2n)]$. Note that changing the $t$-values by $\epsilon/n$, changes the boundary of task $p_k$ by at most $\epsilon(n-1)/n<\epsilon$, so task $p_k$ is allocated to the 0-player in $\hat t^1$, but not in $\hat t^2$.

  Let $a^1$ and $a^2$ be the allocations for $\hat T_1$ and $\hat T_2$. By weak monotonicity we have
  \begin{align*}
    \sum_{i\in [k-1]}(a_{p_i}^2-a_{p_i}^1)(\hat t_{p_i}^2-\hat t_{p_i}^1)+(0-1)(t_{p_k}^2-t_{p_k}^1)&\leq 0 \\
    \sum_{i\in [k-1]}(a_{p_i}^2-a_{p_i}^1)(\hat t_{p_i}^2-\hat t_{p_i}^1) &\leq \delta^2-\delta^1+2\epsilon \\
    \sum_{i\in [k-1]}(a_{p_i}^2-a_{p_i}^1)(\hat t_{p_i}^2-\hat t_{p_i}^1) &\leq -|t_{-p_k}^2-t_{-p_k}^1|_1+2\epsilon.
  \end{align*}
  By the Lipschitz property when $\epsilon$ is arbitrarily small, the last inequality can be satisfied only when $a_{p_i}^2-a_{p_i}^1=-1$ for every $i\in [k-1]$. In particular we have $a_{p_i}^2=0$, which means that in the allocation of $\hat T_2$, no task in $P$ is allocated to the 0-player. Since $\hat t_{p_i}^2<t_{p_i}^2$, $i\in [k-1]$, and $\hat t_{p_k}^2=\delta^2+\epsilon<\delta$, by weak monotonicity, still no task in $P$ is allocated to the 0-player for values $t_{p_i}^2$ and $t_{p_k}=\delta$.

  In summary, no task in $P$ is allocated to the 0-player for the instance that results from $T_2$ when we change the values of task $p_k$ to $[t_{p_k}=\delta,\, s_{p_k}=q^* \delta/(2n)]$. But this contradicts the assumption that $P$ is potentially good.
\end{proof}

We can use this stronger lemma to provide a direct proof of
Lemma~\ref{lemma:prob-of-linear}, which avoids the use of volume and instead of
showing that instances with linear boundary functions have probability 0, it
shows that no such instance exists at all.

\begin{lemma} \label{lemma:prob-of-linear-strong} Suppose that $P=(p_1,\ldots,p_k)$
  is a potentially-good but not good set of tasks for $T$ with witness
  $V$. There exists a witness $V'$ such that no instance $T'$ in the set of
  $V'$-perturbations of $\hat T(P)$ has linear boundary function
  $\psi_{p_k}(s_{p_k})$ (when the values of the other tasks are fixed to the
  values of $T'$).
\end{lemma}
\begin{proof}
  Towards a contradiction, suppose that there exists some instance $T'=(t',s)$
  with linear boundary function $\psi_{p_k}(s_{p_k})$. By
  Lemma~\ref{lemma:magic-affine minimizer - strong}, for no other instance
  $T''$, with $T''\leq T'$, the boundary function $\psi_{p_k}(s_{p_k})$ is
  linear. That is, there is no instance with linear boundary function in the
  witness defined by intervals $V_{p_i}'=(\alpha, t_{p_i}')$, $i\in [k]$.
\end{proof}

\subsubsection{Relaxed affine minimizers}
\label{sec:relax-affine-minim}

Lemma~\ref{lemma:prob-of-linear} essentially excludes the case of
$2\times 2$ affine minimizers. %
We now want to deal with the case of relaxed affine
minimizers. This presents an additional difficulty as we may not be in
the linear part of the mechanism but at its tail. Although this seems
like a special case, it adds another complication to the proof and we
need to handle it carefully. We deal with this problem by showing that
for most siblings $p_k'$, the $2\times 2$ mechanism cannot be a
relaxed affine minimizer, since then it will be in its linear part and
it will remain in the linear part when we change the $s_{p_k}$ value,
and hence we can treat it like in Lemma~\ref{lemma:magic-affine
  minimizer}. To achieve this, we first show that not many tasks can
be allocated to the $k$ player, otherwise the approximation ratio is
high (Lemma~\ref{lemma:n-tasks-to-s} below).  This is useful, because
when a task $p_k'$ is allocated to the 0-player, and the
$(p_k,p_k')$-slice mechanism is a relaxed affine minimizer, we can
guarantee that it is in the linear part of the mechanism and it will
remain there when we change $s_{p_k}$ (see
Lemma~\ref{lemma:2x2-relaxed-aff-min} below).

\begin{lemma} \label{lemma:n-tasks-to-s} Suppose that $P=(p_1,\ldots,p_k)$ is a potentially-good set of tasks for $T$ with witness $V$ and that the approximation ratio is less than $\rho$. For every instance of the set of $V$-perturbations of $\hat T(P)$, fewer than $2n^2$ siblings of $p_k$ are allocated to the $k$-player.
\end{lemma}
\begin{proof}
  By the definition of $T$, we have $[t_{p_k'}=\beta, \,
  s_{p_k'}=1]$. Fix an instance $T'=(t',s)$ of the set of
  $V$-perturbations of $\hat T(P)$. If the mechanism allocates $2n^2$
  or more sibling tasks $p_k'$ to $k$-player, the mechanism has
  approximation ratio at least $\rho$ which leads to a contradiction:
  the approximation ratio is at least $\rho$ because the mechanism has
  cost at least $2n^2$ (from the sibling tasks), while the optimum
  allocation has cost at most $k\alpha+(n-1)\delta'\leq
  n\alpha+n\delta'\approx \sqrt{n}$, and the ratio is at least $\rho$
  (for $n>1$).
\end{proof}

In the above lemma, we can easily improve the number of tasks from
$2n^2$ to $O(n)$, by a simple tightening of its proof, but we again
opt for simplicity.

\begin{lemma} \label{lemma:2x2-relaxed-aff-min} Fix an instance $T$ with $t_{p'}>0$, two tasks $p$ and $p'$ of the same cluster and assume that the $(p,p')$-slice mechanism for $T$ is a relaxed affine minimizer. If the mechanism allocates task $p'$ to 0-player but not task $p$, then this instance is not in the bundling tail of the relaxed affine minimizer, and the same holds for every value of $s_p$. That is, the boundary of $\psi_p(s_p)$ is linear for any value of $s_p$, for fixed values for $t_{-p}$ and $s_{-p}$.
\end{lemma}

\begin{proof}
  A relaxed affine minimizer can be in its non-linear part (bundling tail), only if for every positive values of $t_p$, $t_{p'}$, both tasks are given to the 0-player or none. Since the mechanism allocates exactly one of tasks $p$ and $p'$ to 0-player, the mechanism is not in the bundling tail. Also, the fact that the mechanism allocates task $p'$ but not task $p$ to the 0-player implies that $\psi_{p'}(s_{p'})>0$ (see Figure~\ref{fig:notbundling}). %
   By the definition of relaxed affine minimizers, the rectilinear parts ---horizontal parts in the figure--- of the boundary $\psi_{p'}(s_{p'})$ do not depend on $s_p$ and therefore they remain the same when we change $s_p$. This prevents the allocation to become bundling.
  \begin{figure}
    \centering
    \begin{tikzpicture}

      \draw[->] (0,0) -- (6,0) node[anchor=north] {$t_{p}$};
      \draw[->] (0,0) -- (0,5.4) node[anchor=east] {$t_{p'}$};
      \draw[very thick, blue] (1.0,4.58) -- (4,1.58) -- (6,1.58) ; \draw[ultra thick, red, dashed] (1.0,4.7) -- (4.1,1.6) ; \draw[ultra thick, red, dashed] (4.1,1.6) -- (4.1,0) node[anchor = north] {$\psi_p(s_p)$};

      \draw[ultra thick, red, dashed] (0,3.7) -- (2.1,1.6) ; \draw[ultra thick, red, dashed] (2.1,1.6) -- (2.1,0) node[anchor = north] {$\psi_p(s'_p)$};

      \draw (1.5,0.8) node[anchor=east] {$R_{p,p'}$}; \draw (5.5,0.8) node[anchor=east] {$R_{p'}$}; \draw (4,4) node[anchor=east] {$R_{\emptyset}$}; \draw[ blue, dashed] (0,1.58) node[anchor=east] {$\mathcal{P}_{p'}$} -- (6,1.58) ;

    \end{tikzpicture}

    \caption{\small The figure shows the allocation of the 0-player of a $(p,p')$-slice mechanism which is relaxed affine minimizer. If the mechanism gives task $p'$ to the $t$-player, but not task $p$, then it must be that $\mathcal{P}_{p'}>0$. When we change the value of $s_p$ to $s'_p$, the boundary $\psi_p(s'_p)$ slides left and the horizontal boundary remains at the same level $\mathcal{P}_{p'}$. Therefore the allocation cannot become bundling by changing only $s_p.$}
    \label{fig:notbundling}
  \end{figure}
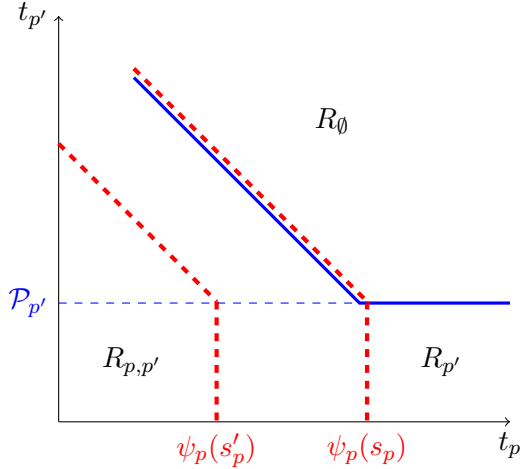
\end{proof}

One of the premises of the previous lemma is that $t_{p'}>0$. We will apply this lemma when $t_{p'}=\beta$. 

We now use a probabilistic argument that combines the above lemmas to show that the probability of relaxed affine minimizers is bounded by $2n^2/\ell$.

\begin{lemma} \label{lemma:prob-of-relaxed-aff-min} Suppose that $P=(p_1,\ldots,p_k)$ is a potentially-good but not good set of tasks for $T$ with witness $V$. Take a random instance $T'$ from the set of $V$-perturbations of $\hat T(P)$ and a random sibling $p_k'$ of $p_k$. We consider the $(p_k,p_k')$-slice mechanism of tasks $p_k$ and $p_k'$ when we fix the values according to $T'$ for the other tasks. The probability that this $2\times 2$ mechanism is a relaxed affine minimizer or affine minimizer is at most $2n^2/\ell$.
\end{lemma}
\begin{proof}
  Let $\cal E$ be the event that the $(p_k,p_k')$-slice mechanism for $T'$ is affine minimizer or relaxed affine minimizer, and let $\cal F$ be the event that task $p_k'$ is allocated to the 0-player. We can trivially bound the probability of event $\cal E$ as follows:
  \begin{align*}
    \Pr[{\cal E}] &\leq \Pr[{\cal E}\cap {\cal F}]+\Pr[\overline{\cal F}]
  \end{align*}
  The probability that event $\overline{\cal F}$ happens is at most $2n^2/\ell$, by Lemma~\ref{lemma:n-tasks-to-s}. To bound the other term, observe that when task $p_k'$ is given to the 0-player, we are in the case of linear $(p_k,p_k')$-slice mechanisms that are linear and remain linear even when we change $s_{p_k}$ (Lemma~\ref{lemma:2x2-relaxed-aff-min}). But the probability that the $(p_k,p_k')$-slice mechanisms is linear is 0 (Lemma~\ref{lemma:prob-of-linear}). Therefore $\Pr[{\cal E}]\leq \Pr[\overline{\cal F}]\leq 2n^2/\ell$.
\end{proof}

\subsubsection{1-dimensional and constant mechanisms}
\label{sec:1-dim-constant}

Here, we show that 1-dimensional and constant slice mechanisms can
be immediately excluded.

\begin{lemma} \label{lemma:prob-of-1dim} Suppose that $P=(p_1,\ldots,p_k)$ is a potentially-good set of tasks for $T$ with witness $V$. Fix a sibling $p_k'$ of $p_k$ and some instance $T'$ in the set of $V$-perturbations of $\hat T(P)$. Suppose that the $(p_k,p_k')$-slice mechanism for $T$ is a 1-dimensional or a constant mechanism. Then the mechanism has approximation ratio at least $\rho$.
\end{lemma}

\begin{proof}
  If the $(p_k,p_k')$-slice mechanism for $T$ is 1-dimensional, there is at least one area missing from the allocation of tasks $p_k,p'_k$. We will show a high ratio by changing only the values of tasks $p_k, p'_k$. If the area that assigns both tasks to the $k$-player is missing, change the values to $[t_{p_k}=n^3, \, s_{p_k}=1]$ and $[t_{p'_k}=n^3, \, s_{p'_k}=1]$. Then the 0-player will be assigned at least one task with cost at least $n^3$, while in the optimum allocation it gets none of them and the cost is at most $(n-1)(1+\delta')$.  The case where the area that assigns both tasks to the $0$-player is missing, is similar.

  The only remaining case is when the mechanism is bundling. In that case take values $[t_{p_k}=n^3, \, s_{p_k}=\beta]$ and $[t_{p'_k}=\beta, \, s_{p'_k}=n^3]$. Then the only allocation with small makespan is the one that assigns task $p_k$ to the $k$-player and task $p'_k$ to the 0-player, but this allocation is not considered by any bundling mechanism.

  For constant mechanisms, the argument of Lemma~\ref{lemma:2x2-aff-min}, shows that the approximation ratio is high.
\end{proof}

\subsubsection{Task independent and relaxed task independent mechanisms}
\label{sec:indep-or-relax}

In this subsection we deal with task independent and relaxed task
independent slice mechanisms. The next lemma is analogous to
Lemma~\ref{lemma:magic-affine minimizer} for task independent or
relaxed task independent slice mechanisms. There is another layer of
difficulty in this case, since we need to take $k$ instances instead
of just two instances that we considered in the linear case. The
difference is that the goal of Lemma~\ref{lemma:magic-affine
  minimizer} was to essentially exclude linear mechanisms, but the
goal of the next lemma is to show that task independent slice
mechanisms imply abundance of good sets of tasks.

\begin{lemma} \label{lemma:magic-task-independent} Suppose that $P=(p_1,\ldots,p_k)$ is a potentially-good but not good set of tasks for $T$ with witness $V$. Fix a sibling $p_k'$ of $p_k$. Suppose that there exists instances $T_0$, $T_1$,\ldots, $T_{k-1}$ such that for every $i\in \{0,\ldots,k-1\}$
  \begin{itemize}
  \item $T_i$ belongs to the set of $V$-perturbations of $\hat T(P)$
  \item for every $i\neq 0$, $T_i$ differs from $T_0$ only in the value of $t_{p_i}$; furthermore, the value of $T_i$ for this task is higher than the value of $T_0$
  \item the $(p_k,p'_k)$-slice mechanism for $T_i$ is a task independent or a relaxed task independent mechanism.
  \end{itemize}
  Assuming that there is no bad task and that approximation ratio is less that $\rho$, the set of tasks $(P_{-k},p_k')$ is good for $T$.
\end{lemma}

\begin{proof}
  Take the instances $T_0$ and $T_i$, for some $i\in [k-1]$. The two
  instances agree on all values except for their $t$-values of task
  $p_i$. Let $\alpha'$ and $\alpha''$ denote these two different
  values for $T_0$ and $T_i$, respectively, with
  $\alpha'<\alpha''$. Let also $\alpha_k=t_{p_k}$.
  We assume that both $(p_k,p_k')$-slice mechanisms for $T_0$ and
  $T_i$ are task independent and at the end of the proof we show how
  to adapt the argument to deal with relaxed task independent
  mechanisms.

  We first establish that when we change the value of $t_{p_k}$ in
  $T_0$ to $\beta$ and the value of $t_{p_k'}$ to $\alpha+\epsilon$,
  for some $\epsilon\in(0,\beta)$, all tasks in $P_{-k}$ are allocated
  to the 0-player.
  To do this, we focus on the $(p_i,p_k)$-slice mechanism for $T_0$ and observe (Lemma~\ref{lemma:all-tasks-to-s}) that there is a bundling boundary between tasks $p_i,p_k$ when $t_{p_i}\in (\alpha',\alpha'')$ (see Figure~\ref{fig:magic-task-independent}). 
  Now, the point is that since the $(p_k,p_k')$-slice mechanisms are task independent, the boundary of task $p_k$ does not change when we change $t_{p_k}$ and $t_{p_k}'$, and because it is tightly connected with the bundling boundary of task $p_i$, the boundary of $p_i$ does not change either (see Figure~\ref{fig:magic-task-independent}). In particular, if we change $T_0$ so that $t_{p_k}=\beta$ and $t_{p_k'}=\alpha+\epsilon$, then task $p_k$ is allocated to the 0-player, and because it is bundled with $p_i$, task $p_i$ is also allocated to the 0-player. This is true for all $p_i\in P_{-k}$, so we established that all tasks in $P_{-k}$ are allocated to the 0-player.

\begin{figure}
  \centering
  \begin{tikzpicture}

    \draw[->] (0,0) -- (6,0) node[anchor=north] {$t_{p_k}$};
    \draw[->] (0,0) -- (0,6) node[anchor=east] {$t_{p_i}$};
    \draw[very thick, blue] (0,3.5) -- (2,3.5) -- (4,1.5) -- (5,1.5) node[anchor=west] {$\psi_{p_i}(t_{p_k})$} ;

    \draw[ blue, dashed] (0,2.4) node[anchor=east] {$\alpha'$} -- (5.6,2.4) ;

    \draw[ blue, dashed] (0,3.3) node[anchor=east] {$\alpha''$} -- (5.6,3.3) ;

    \draw[ blue, dashed] (2.1,3.3) node[anchor={south west}] {$\psi_{p_k}(\alpha'')$} ;

    \draw[ blue, dashed] (2.9,2.4) node[anchor={south west}] {$\psi_{p_k}(\alpha')$} ;

    \draw[ blue, dashed] (3.4,0) node[anchor=north] {$\alpha_k$} -- (3.4,5.6) ; \draw[ blue, dashed] (0.2,0) node[anchor=north] {$\beta$} -- (0.2,6) ;

    \draw[very thick, blue] (4,0) -- (4, 1.58) -- (2, 3.58) -- (2, 5.5) node[anchor={south}] {$\psi_{p_k}(t_{p_i})$} ;

    \draw[ultra thick,black] (2.2,3.3) circle (1mm); \draw[ultra thick,black] (3.1,2.42) circle (1mm);

    \draw[ultra thick, red] (3.1, 2.42) -- (2.2, 3.3) (1, 4) %
    ;

  \end{tikzpicture}

  \caption{\small Illustration for the proof of Lemma~\ref{lemma:magic-task-independent}. The blue figure shows the allocation of the $0$-player for tasks $p_{i},p_{k}$, when the values of task $p'_k$ are as in $T_0$. The black circles show the bundling boundary points $\psi_{p_k}(\alpha',1)$, $\psi_{p_k}(\alpha'',1)$. By independence, the $\psi_{p_k}$ boundary does not change when $t_{p_k'}$ changes. But since the boundary functions $\psi_{p_k}$ and $\psi_{p_i}$ must agree on the red part of $\psi_{p_k}$, this part remains the same for both $\psi_{p_k}$ and $\psi_{p_i}$. As a result, when $t_{p_i}=\alpha'$, $t_{p_k}=\beta$ and $t_{p_k'}=\alpha+\epsilon$ (not shown in this figure), both tasks $p_i$ and $p_k$ are given to the $0$-player.}
  \label{fig:magic-task-independent}
\end{figure}
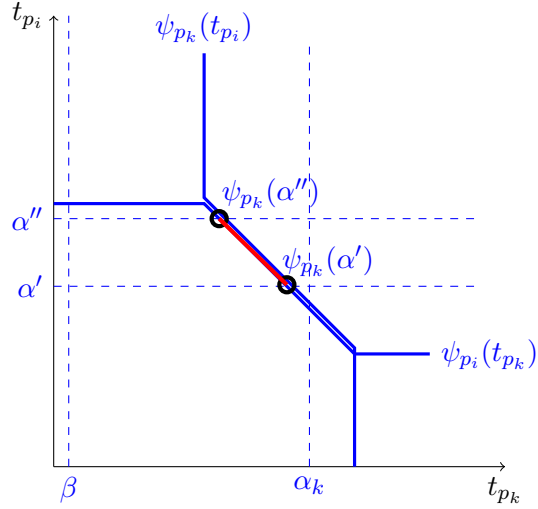

Next we argue that task $p_k'$ is also assigned to the 0-player,
otherwise $\{p_k'\}$ is a bad task, a contradiction. Indeed if task
$p_k'$ is assigned to the $k$-player, we can set $t_{p_i}=\beta$, for
each task $p_i$ in $P_{-k}$ and increase slightly the $t$-value of task
$p_k'$. By weak monotonicity, the allocation of the tasks in $P$ will
remain the same.

We have established that all tasks in $(P_{-k},p_k')$ have values in $(\alpha,
\alpha+\beta)$ and are allocated to the 0-player. By
Lemma~\ref{lemma:new-witness}, the set $(P_{-k},p_k')$ is a good set
for $T$.

The case of relaxed task independent mechanisms is not really different. Recall that a relaxed task independent mechanism behaves like a task independent mechanism in all but countably many values of $s'_{p_k}$. Therefore, there is some arbitrarily small $\epsilon>0$ and some value of $s_{p'_k}$ in $(1-\epsilon,1]$ for which the mechanism is task independent and the above argument works\footnote{This is the only place where the definition of a bad task includes also the case of an $s$-value $1-\epsilon$. }. In order to establish that $P_{-k}+\{p'_k\}$ is good we need to reset the value of $s_{p'_k}$ back to $1$. When we do this, by the weak monotonicity property for player $k$, task $p'_k$ is still allocated to the $0$-player. The allocation for tasks in $P_{-k}$ should also remain the same, for otherwise, we can apply weak monotonicity for the 0-player to contradict the assumption that $P_{-k}$ is good as follows. The trick is the same, by changing the values of tasks in $(P_{-k},p_k')$: lowering to $\beta$ the $t$-value of $p'_k$ and increasing slightly the $t$-values of all tasks in $P$ that are not allocated to the 0-player.
\end{proof}

Taking advantage of the above lemma, we now show that either a
potentially-good set of tasks is also good or there are many other
good sets.

\begin{lemma} \label{lemma:many-good-sets} Suppose that
  $P=(p_1,\ldots,p_k)$ is a potentially-good set of tasks for $T$ with
  witness $V$. Then either $P$ is good for $T$ or the probability that
  $(P_{-k},p_k')$ is good for $T$, for a random sibling $p_k'$ of
  $p_k$, is at least $1-2n^3/\ell$.
\end{lemma}
\begin{proof}
  Suppose that $P$ is not good. Take a random sibling $p_k'$ of $p_k$,
  and random instances $T_i$, $i=0,\ldots,k-1$ from $V$ such that
  $T_0$ and $T_i$ differ in only the value of $p_i$, as in
  Lemma~\ref{lemma:magic-task-independent}. Let ${\cal E}_{i,p_k'}$ be
  the event that the $(p_k,p_k')$-slice mechanism is not (relaxed)
  task independent for $T_i$, $i=0,\ldots,k-1$. By combining Lemmas
  \ref{lemma:prob-of-relaxed-aff-min} and \ref{lemma:prob-of-1dim},
  the probability of each event ${\cal E}_{i,p_k'}$ is at most
  $2n^2/\ell$, so the probability of the event $\cup_{i=0,\ldots,k-1}
  {\cal E}_{i,p_k'}$ is at most $k \cdot 2n^2/\ell\leq
  2n^3/\ell$. Therefore with probability at least $1-2n^3/\ell$, all
  these $(p_k,p_k')$-slice mechanisms are task independent or relaxed
  task independent and therefore $(P_{-k},p_k')$ is good.

\end{proof}

\subsection{Existence of a good set of tasks of size $n-1$}
\label{sec:sum-up}

The above analysis culminated in Lemma~\ref{lemma:many-good-sets}. In this
subsection, we use this lemma to show the existence of a good set of tasks of
size $n-1$ (Corollary~\ref{cor:ell}), which immediately concludes the proof of
the Main Lemma (Lemma~\ref{lemma:main}). We will use a probabilistic argument
again, so we first define the probability $b_k$ that a regular set of $k$ tasks
is bad.
\begin{definition}[Probability $b_k$]
  Fix a mechanism and suppose that $T$ is a standard instance for a set $\cal C$ of $k$ clusters. %
  Let us denote by $b(T,{\cal C})$ the probability that a random regular set $P$ of $k$ tasks from $\cal C$ is not good and let $b_k$ denote the maximum possible $b(T,{\cal C})$ ---or more precisely, its supremum--- for all choices of $T$ and $\cal C$, for the given mechanism.
\end{definition}

First, we find an upper bound on the probability that a random regular
set of tasks of size $k$ is not potentially-good, based on the
probability that a set of size $k-1$ is not good.

\begin{lemma} \label{lemma:potentially-good} Fix any instance $T$
  which is standard for a set $\cal C$ of $k\geq 2$ clusters, and let
  $P$ be a random regular set of $k$ tasks from $\cal C$. The
  probability that $P$ is not potentially-good for $T$ is at most
  $(3n/\delta-1) b_{k-1}$.
\end{lemma}
\begin{proof}
  Take a random set $P=(p_1,\ldots,p_k)$ of $k$ regular tasks from
  $\cal C$. By definition, the probability that $P_{-i}$ is not good
  for $T$ is at most $b_{k-1}$, for every $i\in [k]$. Also the probability that $P_{-k}$ is
  not good when we set task $p_k$ to $[t_{p_k}=\delta, \, s_{p_k}=q
  \delta/(2n)]$ is also at most $b_{k-1}$, for each $q=0,1,\ldots,
  2n/\delta$. By the union bound, the probability that some of these
  bad events happens is at most\footnote{The expression $3n/\delta-1$
    may appear arbitrary here, but it is a convenient upper bound to
    simplify the expressions of the next lemma.}
  $(k+2n/\delta+1)b_{k-1}\leq (n+2n/\delta)b_{k-1} \leq (3n/\delta-1)
  b_{k-1}$ (since $\delta \ll 1$).
\end{proof}

The following lemma estimates a rough upper bound on the probability that a random regular set is not good. It shows that if $\ell$ has exponential size in $n$, the probability that there exists a good regular set of $n-1$ tasks is positive.

\begin{lemma}[Probability of bad sets]\label{lemma:prob-of-bad}
  Fix a mechanism with approximation ratio less than $\rho$ and suppose that
  there is no bad task. Then the probabilities of bad events are bounded above
  by
  \begin{align} \label{eq:bk} b_k\leq \left(\frac{3n}{\delta}\right)^{k-2} \frac{3n^3}{\ell},
  \end{align}
  for every $k\in [n-1]$.
\end{lemma}
\begin{proof}
  By the premises of the lemma, there is no bad task, and hence $b_1=0$. We first show that for $k\geq 2$,
  \begin{align} \label{eq:bk-rec} b_k\leq \left(\frac{3n}{\delta} -1\right) b_{k-1}+\frac{3n^3}{\ell}.
  \end{align}

  Fix $k$ distinct clusters $(Q_1,\ldots,Q_k)$ and consider a regular set $P=\{p_1,\ldots, p_{k-1}\}$ with $p_i\in Q_i.$ %
  Let's call the sets $P+\{p_k\}$, $p_k\in Q_k$, extensions of $P$.

  Let $y\in [\ell+1]$ be the number of the extensions of $P$ that are
  potentially-good. How many of these extensions are good? Either all $y$
  extensions are good, in which case the number of good extensions is at least
  $y$ or some extension is not good and $\ell-2n^3$ other extensions are good
  (Lemma~\ref{lemma:many-good-sets}). Therefore the number of extensions that
  are good is at least
  $\min\{y, \ell-2n^3\}\geq y (1-3n^3/(\ell+1)) \geq y (1-3n^3/\ell)$.  In other
  words, the probability that a random regular set of size $k$ is good is at
  least $1-3n^3/\ell$ times the probability that a random regular set of size
  $k$ is potentially-good. By Lemma~\ref{lemma:potentially-good}, the
  probability that a random regular set of size $k$ is potentially-good is at
  least $1-(3n/\delta - 1) b_{k-1}$. Using the union bound we get
  Equation~\eqref{eq:bk-rec}.

  We now use Equation~\eqref{eq:bk-rec} to derive the expressions of the lemma.
  Note first that since $b_1=0$, Equation~\eqref{eq:bk-rec} gives $b_2\leq 3n^3/\ell$.

  For $k\geq 3$, using Equation~\eqref{eq:bk-rec} inductively we get the bound:
  \begin{align*}
    b_k\leq \left(\frac{3n}{\delta}-1\right) b_{k-1}+\frac{3n^3}{\ell} \leq
    \frac{3n^3}{\ell} \left(\frac{3n}{\delta}\right)^{k-2}.
  \end{align*}
  The lemma follows by observing that Equation~\eqref{eq:bk} holds also for
  $k=1$ and $k=2$.
\end{proof}

Using this, we can now reach the target of this section, showing that there exists a good set of size $n-1$.
\begin{corollary} \label{cor:ell} If $\ell > 3n^3 \left(\frac{3n}{\delta}\right)^{n-3}$, we get $b_{n-1} < 1$, and there exists at least one good set of $n-1$ tasks.
\end{corollary}

\subsection{Lower bound on the approximation ratio}
\label{sec:appr-ratio-sqrtn}

In this subsection we use the Main Lemma (Lemma~\ref{lemma:main}) to
prove a lower bound of $1+\sqrt{n-1}$. We also show that this is
almost tight for the instances in which each task can be allocated
either to the 0-player or to some other player
(Definition~\ref{def:range-of-values}).

\begin{theorem}
  The approximation ratio of truthful mechanisms is at least $1+\sqrt{n-1}$.
\end{theorem}
\begin{proof}
  We first show that Lemma~\ref{lemma:main} implies that the approximation ratio is at least equal to
  \begin{align*}
    \rho &=1-\delta'+\min\left\{\frac{(n-1)\alpha}{1+(n-1)\delta'}, \,
           \frac{1}{\alpha+(n-1)\delta'}\right\}.
  \end{align*}
  Indeed if this is not the case, by Lemma~\ref{lemma:main}, either there exists a bad task or a good set of $n-1$ tasks. In the first case, the mechanism allocates the bad task to the $s$-player with cost at least 1. On the other hand, the optimal algorithm allocates the task to the 0-player and has cost at most $\alpha+(n-1)\delta'$, where the first term comes from the bad task and the term $(n-1)\delta'$ comes from the remaining tasks, that are trivial. The approximation ratio is at least $1/(\alpha+(n-1)\delta')$. We can improve the approximation by an additive term ``+1'' using the dummy tasks, as in the proof of Lemma~\ref{lemma:2x2-aff-min}. Indeed, if the bad task is in cluster $C_i$, we change the $s$-value of dummy task $d_i$ to $\alpha+(n-1)\delta'$  and slightly decrease the $s$-value of the bad task. By weak monotonicity of the $i$ player, this does not affect the allocation, and the approximation ratio increases by 1.

  Otherwise there exists a good set $P$ of $n-1$ tasks. By the definition of goodness, all tasks in $P$ are allocated to the 0-player and the cost is at least $(n-1)\alpha$. On the other hand, when all tasks in $P$ are allocated to the $s$-players, the optimal cost is at most $1+(n-1)\delta'$; again the second term comes from the remaining tasks. The approximation ratio is at least $(n-1)\alpha /(1+(n-1)\delta')$. Again, we increase the approximation ratio by 1 using the dummy tasks. In this case, we change the $t$-value of dummy task $d_0$ to $1+(n-1)\delta'$.
  
  By selecting $\alpha=1/\sqrt{n-1}$ the approximation ratio is $\sup_{\delta'>0} \rho=1+\sqrt{n-1}$. Note that although the instances have large number of tasks, for every $\delta'>0$, the number of tasks is finite.
\end{proof}

\begin{remark}
To achieve a ratio very close to $1+\sqrt{n-1}$, we must select a small
$\delta'$. By Corollary~\ref{cor:ell}, this has implications on the number of
tasks required by the proof of the main lemma (Lemma~\ref{lemma:main}). In
particular the number of tasks per cluster must be more than $3n^3
\left(\frac{3n}{\delta}\right)^{n-3}=n^{\Theta(n)} /\delta^{n-3}$. Note that by
any choice of $\delta<1$, this number is at least exponential in $n$. By
selecting $\delta\approx 1/n^2$, we still get the number of tasks to be
exponential in $n$ and the approximation ratio to be $(1-o(1))\sqrt{n-1}$.
\end{remark}

We now show that the lower bound is tight for the instances that we consider.

\begin{theorem}
  The approximation ratio for the set of instances of Definition~\ref{def:range-of-values} is at most $1+\sqrt{n-1}$ and it is achieved by a weighted version of the VCG mechanism.
\end{theorem}
\begin{proof}
  We consider the weighted VCG mechanism by which each task $j$ is allocated to the 0-player if and only if $\lambda \, t_j\leq s_j$, where $\lambda=\sqrt{n-1}$. Let $r_i$ and $l_i$ be the total weight of tasks from cluster $C_i$ allocated by this algorithm to the 0-player and to the $i$-th player, respectively. Let $r_i^*$ and $l_i^*$ be the corresponding quantities for an optimal algorithm. An equivalent way to describe the weighted VCG algorithm is to say that it achieves the minimum $\lambda \,r_i+l_i$ for every cluster $C_i$ among all algorithms.

  The cost of the VCG algorithm is $\mech=\max\{\sum_{i \in [n-1]} r_i, \max_{i\in [n-1]} l_i\}$. If $\mech=\sum_{i \in [n-1]} r_i$, we get

  \begin{align*}
    \mech =\sum_{i \in [n-1]} r_i & \leq \sum_{i \in [n-1]} \left(r_i+\frac{l_i}{\lambda}\right)
    & \text{(and by the definition of the mechanism)} \\
                                  & \leq \sum_{i \in [n-1]} \left(r_i^*+\frac{l_i^*}{\lambda}\right) \\
                                  & \leq \sum_{i \in [n-1]}r_i^*+\frac{1}{\lambda}\sum_{i\in [n-1]} l_i^*\\
                                  & \leq \left(1+\frac{n-1}{\lambda}\right)\opt\\
                                  & =(1+\sqrt{n-1})\opt.
  \end{align*}

  Otherwise, $\mech=l_j$ for some $j\in [n-1]$. In this case, we get
  \begin{align*}
    \mech=l_j &\leq \lambda r_j+l_j
    & \text{(and by the definition of the mechanism)} \\
              & \leq \lambda r_j^*+l_j^*\\
              & \leq (\lambda+1) \max\{r_j^*, l_j^*\}\\
              & \leq (\lambda+1)\max \left\{\sum_{i\in [n-1]} r_i^*, \max_{i\in [n-1]} l_i^*\right\}\\
              & \leq (\lambda+1) \opt \\
              & =(1+\sqrt{n-1})\opt.
  \end{align*}

\end{proof}

\bibliographystyle{plain} 
\bibliography{biblio}

\end{document}